\documentclass[a4paper,10pt,oneside,fleqn,review]{cas-dc}
\usepackage[numbers]{natbib}
\usepackage{amsmath,amssymb,amsfonts}
\usepackage{mathrsfs}
\usepackage{amsthm}
\usepackage{tikz}
\usetikzlibrary{matrix}
\theoremstyle{plain}
\usepackage{pst-node}
\newtheorem{theorem}{Theorem}[section]
\newtheorem{lemma}[theorem]{Lemma}

\newtheorem{problem}{Problem}
\newtheorem{proposition}[theorem]{Proposition}
\newtheorem{corollary}[theorem]{Corollary}
\theoremstyle{definition}
\newtheorem{definition}{Definition}
\theoremstyle{remark}
\allowdisplaybreaks
\newtheorem{remark}{Remark}
\DeclareMathOperator*{\argmin}{arg\,min}
\usepackage{xcolor}
\begin{document}
\shorttitle{Multi-Agent Estimation in RKHS}
\shortauthors{Aneesh Raghavan et~al.}
\title [mode = title]{Collaborative Estimation of Real Valued Function by Two Agents and a Fusion Center with Knowledge Exchange}                      
\tnotemark[1]
\tnotetext[1]{This work was supported in part by Swedish Research Council Distinguished Professor Grant 2017-01078, Knut and Alice Wallenberg Foundation Wallenberg Scholar Grant, and the Swedish Strategic Research Foundation FUSS SUCCESS Grant.}
\author[1]{Aneesh Raghavan}
\cormark[1]
\ead{aneesh@kth.se}
\affiliation[1]{organization={DCS Division},
    addressline={Malvinas Vag 10, KTH, Royal Insitute of Technology}, 
    city={Stockholm},
    postcode={11428}, 
    country={Sweden}}
\author[1]{Karl Henrik Johansson}
\ead{kallej@kth.se}
\cortext[cor1]{Corresponding author}
\begin{abstract}
 We consider a collaborative iterative algorithm with two agents and a fusion center for estimation of a real valued function (or ``model") on the set of real numbers. While the data collected by the agents is private, in every iteration of the algorithm, the models estimated by the agents are uploaded to the fusion center, fused, and, subsequently downloaded by the agents. We consider the estimation spaces at the agents and the fusion center to be Reproducing Kernel Hilbert Spaces (RKHS). Under suitable assumptions on these spaces, we prove that the algorithm is consistent, i.e., there exists a subsequence of the estimated models which converges to a model in the strong topology. To this end, we define estimation operators for the agents, fusion  center, and, for every iteration of the algorithm constructively. We define valid input data sequences, study the asymptotic properties of the norm of the estimation operators, and, find sufficient conditions under which the estimation operator until any iteration is uniformly bounded. Using these results, we prove the existence of an estimation operator for the algorithm which implies the consistency of the considered estimation algorithm. 
\end{abstract}
\maketitle
\section{Introduction}\label{Section 1}
\subsection{Motivation and Related Work}\label{Subsection 1.1}
Advancement of technology has led to development of systems (e.g. robots, vehicles, etc.) with enhanced sensing and computing capabilities. The data collected by the systems are heterogeneous, i.e., the sensors inputs and outputs take values in different sets, the noise distributions of sensors are different etc. The data collected by the agents is used to estimate models of the environment or the agents. Often the set of features describing the model is not precisely known; it is reasonable to assume that there is a finite set of feature spaces to which the model belongs. These challenges have motivated and led to different kinds of learning algorithms as described in the following. 

In distributed estimation, information (a function of measurements collected) is exchanged while in federated learning, models are exchanged. Distributed estimation of environment modeled as a random field has been studied in \cite{cortes2009distributed}, \cite{talarico2013distributed}, and, \cite{wang2019optimum}. In \cite{choi2009distributed}, the problem of simultaneous estimation of a static field and cooperative control of a multi-agent system has been studied. In \cite{predd2009collaborative}, a distributed algorithm to solve regression problems in RKHS has been proposed. Distributed supervised learning has been studied in \cite{carpenter1998distributed}, \cite{chang2017distributed} and \cite{ying2018supervised}. Federated learning is a learning architecture where multiple devices (agents) collaborate on training a model under the coordination of a central aggregator. The learning architecture adheres to (i) privacy of data, (ii) local computing, and, (iii) model transmission, \cite{zhang2021survey}, \cite{nguyen2021federated}. When the feature space is not precisely known, vertical federated learning considers estimation in different feature spaces with model exchange. However this framework has not been studied when the estimation spaces are RKHS.  To this end, we consider a multi-agent algorithm with model exchange to estimate a real valued function, where each agent solves an estimation problem in its own feature space, specifically an RKHS. The motivation for considering RKHS is as follows. 

The utilization of kernel methods for identification originated in the signal processing literature and was then extended to system identification as well. Kernel methods have been applied to digital signal processing \citep{rojo2018digital},  analysis of biomedical signals \citep{camps2007kernel}, remote sensing data analysis \citep{gomez2011review}, etc. For a review of the application of these methods to identification of discrete and continuous time systems we refer to \cite{pillonetto2014kernel}. More recently kernel methods have been used for identification of Koopman operators \citep{khosravi2023representer}, system identification with side information \citep{khosravi2023kernel} and hybrid system identification \citep{pillonetto2016new}. 

Learning of operators has received significant attention in literature. Recent advances in deep neural networks has been applied to learning of operators, \cite{bhan2023operator}, \cite{lu2021learning}, \cite{kovachki2023neural}. However, in iterative learning, i.e.,  the processes of learning a new model given old model and data, learning at every iteration can be viewed as an operator across the space of models. This perspective is investigated  in this paper. Here, we use the terms ``knowledge" and "model" interchangeably. The function space or the hypothesis space for solving the estimation problems are referred to as ``knowledge spaces". Thus, the ``learning operator" takes an element in the knowledge space and a data point and outputs an element in the knowledge space. 

\subsection{Objective and Contributions}\label{Subsection 1.2}
The problem set up is as follows. There are two agents and a fusion center. The two agents sequentially receive data comprising of an independent variable (input) and the corresponding value of a dependent variable (output). Each agent is given a set of features which describe the mapping from the input to the output. The objective of the system is to collaboratively estimate the mapping from the input to the output without exchanging data. Since the agents posses different features, their knowledge spaces are different. For the agents to collaborate they need a common space. We consider the architecture proposed in Figure \ref{Figure 1}. The fusion center aids in the collaboration as the agents do not exchange any models, instead they get uploaded and downloaded to fusion center. 
\begin{figure}
    \centering
    \includegraphics[scale=0.29]{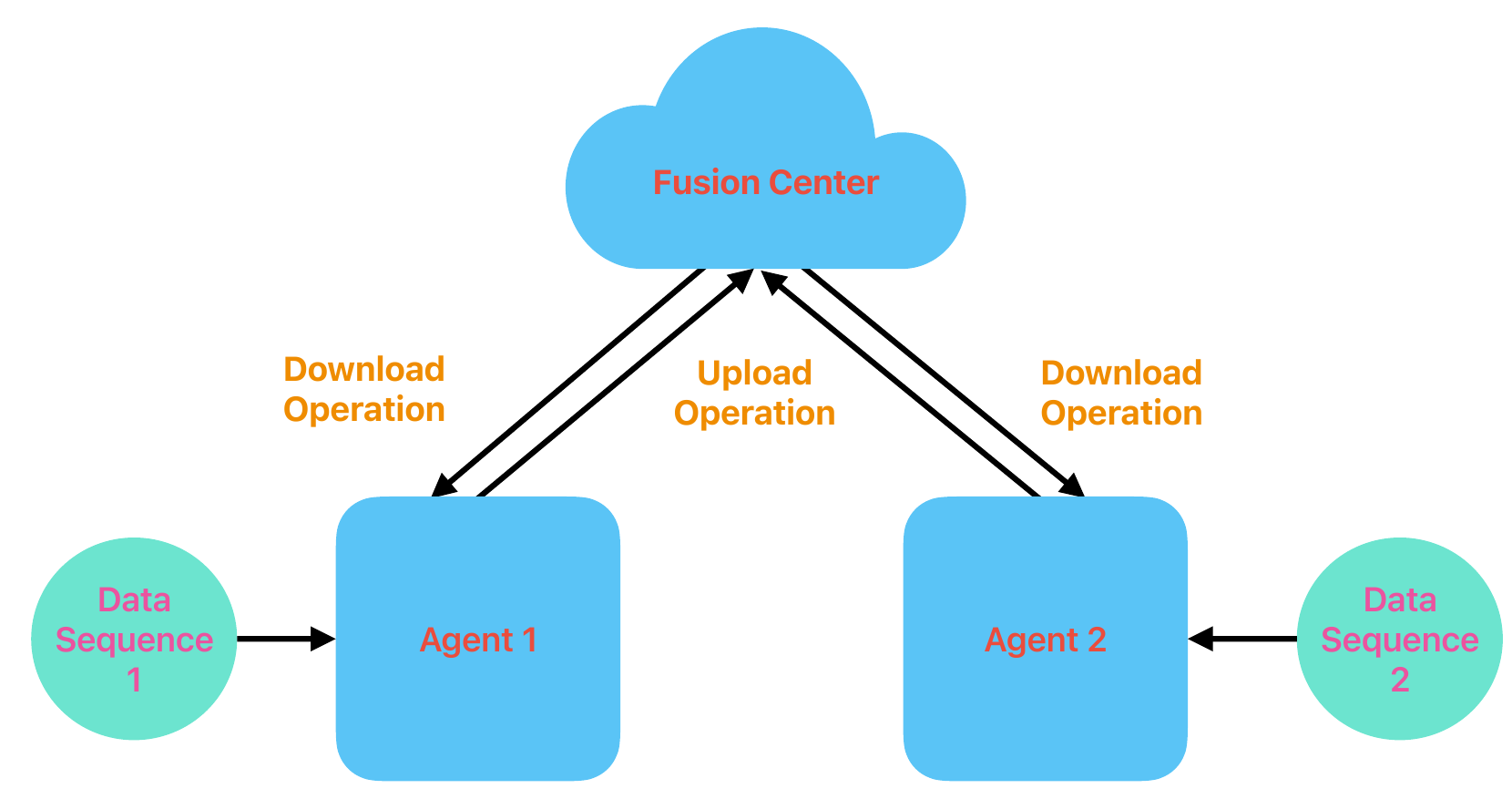}
    \caption{Schematic for the estimation architecture}
    \label{Figure 1}
    \vspace{-0.75cm}
\end{figure}

The algorithm to solve the problem is described as follows. At every iteration $n$, given the current estimate  of the mapping and a data point, each agent obtains a new estimate by solving an estimation problem (specifically a least squares regression problem) in its knowledge space (or hypothesis space, a RKHS). The estimates are uploaded to the fusion center and then fused using a meta-estimation algorithm. The fused function is downloaded on to the individual knowledge space and is declared as the final estimate at iteration $n$. Hence the algorithm generates a sequence of functions in the individual knowledge spaces and the fusion space each of which is a estimate of the mapping from the independent variable to the dependent variable. 

The objective of the paper is to: (i) develop a composable abstraction of the estimation and computations at the agents and the fusion center; (ii) develop an abstraction of the algorithm that enables to characterise its behaviour; (iii) find conditions on the parameters of the estimation problems at the agents, the fusion problem at the fusion center, and, the input data sequences to the algorithm which ensures stable behaviour of the algorithm; (iv) prove that the proposed algorithm is consistent under the conditions in (iii) using the abstract framework developed. 

The contributions of the paper are as follows. We: (i) define estimation operators based on the estimation problems at the agents and fusion center, parameterized by their corresponding  trade-off parameter ; (ii) study the asymptotic properties of these operators; specifically we prove that the norm converges to $1$ as the corresponding parameters diverges to $\infty$; (iii) define the estimation operator for each iteration as the composition of the estimation operator of the agents , an  upload operator (subsection \ref{Subsection 2.1}), the fusion operator at the fusion center and a download operator  (subsection \ref{Subsection 2.1}); (iv) define the input data sequences to the two agents which are suitable for the algorithm; (v) define the estimation operator \textit{until} iteration $n$ as the composition of the estimation operators from iterations $1$ to $n$ and prove that the sequence is an equicontinuous sequence of operators, (vi) using this result, we prove that the algorithm mentioned above is consistent in the sense of Definition \ref{Definition 4} (subsection \ref{Subsection 2.2} ). As a byproduct of the proposed proof methodology, we prove that given a valid  input data sequence, there exists a model in the knowledge space such that it is invariant to the estimation process: the model achieved at the end of the estimation process is the same.  We note the conceptually the algorithm and the proof of consistency presented in this paper can be extended to more than $2$ agents; we restrict ourselves to $2$ for brevity. 
\subsection{Outline}\label{Subsection 1.3}
The organization of the paper is as follows. In Section \ref{Section 2}, we define the estimation problem for the multi-agent system and the analysis problem of proving consistency of the solution of the estimation problem. In Section \ref{ColEst2L- The Estimation Algorithm},  we present an algorithm which solves the estimation problem. In Section \ref{Section 3}, we define operators for the sub-components of the algorithm, study their properties, and, define operators which abstract the iterations of the algorithm. In Section \ref{Section 4}, we present the main result on consistency of the algorithm and the associated corollaries. In Section \ref{Section 5}, we present an example demonstrating the algorithm. We conclude with some comments and future work in Section \ref{Section 6}. In Section \ref{Appendix}, we present the proof of the results utilized in Section \ref{Section 3}.
\section{Problem Formulation}\label{Section 2}
In this section, we formalize the problem setup described in subsection \ref{Subsection 1.2}. In subsection \ref{Subsection 2.1}, we describe the knowledge spaces, the upload, and, the download operators for the agents and the estimation problem for the system.  In subsection \ref{Subsection 2.2}, we define consistency for an algorithm and define the problems with respect to analysis of the algorithm.
\subsection{The Estimation Problem}\label{Subsection 2.1}
The estimation system comprises of two agents, Agent $1$ and Agent $2$, and a fusion center. 
The sequence of data points collected by Agent $i$ is denoted by $\{(x^{i}_n, y^{i}_n )\}$. The set of features for Agent $i$ is a set of continuous functions, $\{\varphi^{i}_{j}(\cdot)\}_{j \in \mathcal{I}^i}$, where $\varphi^{i}_{j} : \mathcal{X} \to \mathbb{R}$, $\mathcal{X} \subset \mathbb{R}^{d}$, and,  $ \mathcal{I}^i = \{1, \ldots, |\mathcal{I}^i| \},  | \mathcal{I}^i | < \infty$. We assume that $\mathcal{X}$ is a closed, connected set (i.e., with no isolated points).The set of features for an agent is assumed to be linearly independent. Let $K^{i}(x,y) = \sum_{j \in \mathcal{I}^i}\varphi^{i}_{j}(x)\varphi^{i}_{j}(y)$ be a kernel and the corresponding RKHS generated by it be $(H^i, \langle \cdot, \cdot \rangle_{H^{i}}, ||\cdot ||_{ H^{i}})$. Then, the knowledge space constructed for Agent $i$ is the RKHS, $H^i$, with kernel $K^{i}$. 

Given the knowledge spaces $H^1$ and $H^2$ with kernels $K^{1}$ and $K^{2}$, the fusion space is a RKHS with kernel $K^{1} + K^{2}$, Lemma \ref{Lemma 27}. A function in the space $H^{i}$ is uploaded or communicated to $H$ using the upload operator $\hat{L}^{i}$, Corollary \ref{Corollary 28}.  A function in $H$ is downloaded or communicated to $H^{i}$ using the download operator $\sqrt{\bar{L}^{i}} \circ \Pi_{\mathcal{N}\big(\sqrt{\bar{L}^i}\big)^{\perp}}$, Lemma \ref{Lemma 30}. The upload and download operators are refereed to as \textit{transfer} operators. We refer to the set of estimation spaces and the transfer operators, $\Big(H^{1}, H^{2}, H, \hat{L}^{i} , \sqrt{\bar{L}^{i}} \circ \Pi_{\mathcal{N}\big(\sqrt{\bar{L}^i}\big)^{\perp}} \Big) $, as the estimation architecture. 
\begin{remark}\label{Remark 1}
    We note that since the set of features for each agent is finite, it follows that the RKHS $H^{1}$ and $H^{2}$ are finite dimensional.
\end{remark}
\begin{problem}\label{Problem 1}
    Given the data points at the agents and the estimation architecture, the estimation problem is to find an iterative collaborative algorithm to find the mapping from the input to the output without any data exchange among the agents.  
\end{problem}
\subsection{The Analysis Problems}\label{Subsection 2.2}
We define consistency of a estimation scheme as follows,
\begin{definition}\label{Definition 4}
Given a pair of initial estimates $(f^{1}_0;f^{2}_0)$ and a suitable sequence of data points $(\hspace{-1pt}\{(x^{1}_n; y^{1}_{n})\}, \{(x^{2}_n; y^{2}_{n})\})$, let $\{(f^{1}_{n}; f^{2}_n)\}$ be the sequence of learned functions in the product space of the individual knowledge spaces of the agents. A estimation scheme is said to be strongly consistent if there exists a subsequence that converges to a function in the product space (equipped with the product of the strong topologies) for every sequence of learned functions, i.e, there exists an infinite index set $I$ such that,  $\big\{f^{1}_{k}; f^{2}_{k}\big\}_{ k \in I} \xrightarrow[ ]{  H^{1} \times H^2} \big(f^{1,*}; f^{2,*} \big)$ for every sequence of learned functions  $ \{f^{1}_{n}; f^{2}_{n}\} $. 
\end{definition}
We refer to $(f^{1}_0;f^{2}_0)$ as initial estimates or initial conditions interchangeably. We recall that a subsequence, $\{ a_k\}_{k \in I}$ of a sequence $\{ a_n\}_{n \in \mathbb{N}}$ is a sequence whose index set, $I \subseteq \mathbb{N}$, has infinite cardinality. We use the notation  $\{ a_k\}_{k \in I}$ and  $\{ a_{n_k}\}$ interchangeably. Usually for strong consistency, it is required that the entire sequence of learned functions is required to converge in the strong topology. Instead, in the above definition, we require that there is a common subsequence that converges for every sequence of learned functions. If such a subsequence exists, with suitable re-indexing of the sequence, then it can be claimed that the ``entire" sequence converges.  The reasoning for the this definition is that this could impose ``weaker" conditions on the algorithm than what would have been required to achieve the usual definition of consistency.

The usual method to prove existence of a converging subsequence for a given sequence is to bound the sequence. However, to prove consistency in the sense of Definition \ref{Definition 4}, we need to \textit{uniformly} bound $|| \{f^{1}_{n}; f^{2}_{n})\} ||$, i.e, the bound should be the same irrespective of the initial conditions and the input data sequence. By defining operators which take in the initial conditions and data as inputs and output the estimates, the uniform bound can be achieved  by uniformly bounding the norm of the operators. 

First, we observe that $(f^{1}_{n}; f^{2}_{n})$ is only dependent on $(f^{1}_{0}; f^{2}_{0})$ and the data points $\{(x^{1}_j; y^{1}_j), (x^{2}_j; y^{2}_j)\}^{n}_{j=1}$. Indeed, the sequence $\{f^{1}_{n}; f^{2}_{n}\}$ can be viewed as the states of a nonautonomous dynamical system with inital conditions  $(f^{1}_{0}; f^{2}_{0})$, and, inputs $\{(x^{1}_j; y^{1}_j), (x^{2}_j; y^{2}_j)\}$. Hence, we would like to find an operator which when given these as inputs, outputs the estimates, $(f^{1}_{n}; f^{2}_{n})$. If such an operator exists, we refer to it as estimation operator until iteration $n$. This operator can be refined further as follows. Next we observe that, given the estimates at $n-1$, the estimates at iteration $n$ is only dependent on the data points at $n$. Hence, we would like to find an operator which when given $(f^{1}_{n-1}; f^{2}_{n-1})$ and $\Big((x^{1}_n; y^{1}_n), (x^{2}_n; y^{2}_n)\Big)$, outputs $(f^{1}_{n}; f^{2}_{n})$. If such an operator exists, we refer to it as the estimation operator at iteration $n$. The estimation operator until iteration $n$ is obtained by composing the estimation operator from iteration $1$ to $n$.

If the operators exist and are uniformly bounded, then $\{(f^{1}_n;f^2_n)\}$, can be uniformly bounded. Using the operator based approach, a common subsequence which converges across all initial conditions and input data sequences could be found. Thus, sufficient conditions under which the norm of the operators is uniformly bounded needs to investigated. Since the input data sequence is an exogenous input, it is not necessarily that is ``well behaved", i.e, the input data sequence may be unbounded causing the algorithm to become unstable. Hence, \textit{valid} input data sequences which lead to consistency of estimates need to be studied.

Following the reasoning as above, we define the problems to analyze the solution of Problem \ref{Problem 1}. 
\begin{problem} \label{Problem 5}
For the solution of Problem \ref{Problem 1}, find estimation operators, $\{\mathbb{T}_n\}$, which abstract iteration $n$ of the algorithm and operators, $\{\bar{\mathbb{T}}_n\}$, which abstract the algorithm until iteration $n$, $n \in \mathbb{N}$.
\end{problem}
\begin{problem}\label{Problem 6}
    Find valid input data sequences which lead to uniformly bounded estimation operators. Find sufficient conditions on norm of the estimation operators, $\{\bar{\mathbb{T}}_n\}$, which guarantees the existence of an operator which abstracts the algorithm for arbitrarily  large $n$. Prove that the solution of Problem 1 is consistent in sense of Definition \ref{Definition 4}. 
\end{problem}
The algorithm which solves Problem \ref{Problem 1} would posses a stopping criterion (subsection \ref{Algorithm Description}) to ensure that the algorithm stops in finite time in practice. However, the analysis presented is an asymptotic study; we investigate convergence of the algorithm where the number of samples could be arbitrarily large.
\section{ColEst2L- The Estimation Algorithm}\label{ColEst2L- The Estimation Algorithm}
In this section, we present an algorithm to solve Problem \ref{Problem 1}. In subsection \ref{Estimation at Agents}, we present the estimation problems at the agents and their solutions and in subsection \ref{Fusion at Fusion Center}, we present the fusion problem and its solution. 
\subsection{Algorithm Description}\label{Algorithm Description}
\begin{figure}
    \centering
    \includegraphics[scale=0.29]{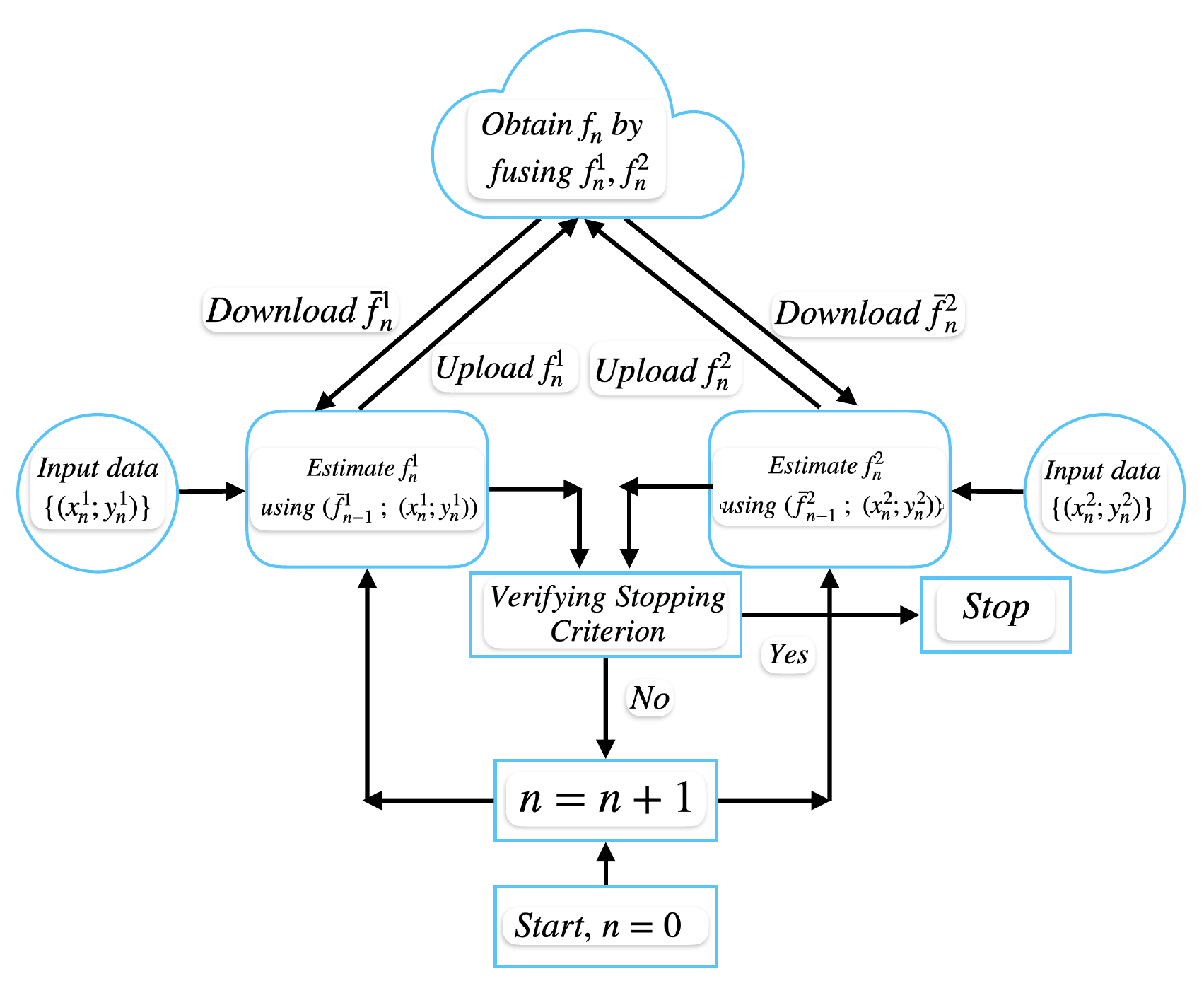}
    \caption{ColEst2L - the estimation algorithm}
    \label{Figure 2}
      \vspace{-0.9cm}
\end{figure}
The multi-agent estimation algorithm, \textit{ColEst2L}, has been described in the flowchart in Figure \ref{Figure 2}. ``ColEst" stands for collaborative estimation while ``2L" stands for 2 levels, level 1 being the agents and level 2 being the fusion center. Given the downloaded function from iteration $n-1$, $\bar{f}^i_{n-1}$ at Agent $i$ and the data point $(x^{i}_n, y^{i}_n)$ at iteration $n$, each agent solves the estimation problem (subsection \ref{Estimation at Agents}) to arrive at the estimate $f^{i}_{n}$. The locally estimated functions, $f^{1}_n$ and $f^{2}_n$, are uploaded to the fusion space. The uploaded functions are fused in the fusion space to obtain $f_{n}$, subsection \ref{Fusion at Fusion Center} . The fused function  downloaded onto the knowledge space of Agent $i$, denoted as $\bar{f}^i_{n}$, is considered as the final estimate at the agents at iteration $n$. For a given $k_{\max} \in \mathbb{N}$ and $\epsilon >0$, it is verified if $|| \bar{f}^{1}_{n - k_{\max} +j} -  \bar{f}^{1}_{n - k_{\max}}|| + || \bar{f}^{2}_{n - k_{\max} +j} -  \bar{f}^{2}_{n - k_{\max}}|| < \epsilon $, for every $ j \in \{0,1, \ldots, k_{\max}\}$. If the condition is satisfied, then the algorithm is terminated. Else $n$ is incremented by $1$ and the procedure is repeated. The condition verifies if the downloaded functions at the agents are ``close" enough in the norm sense. This collaborative estimation algorithm and some of its properties, specifically, the upload and download operation were studied in \cite{Raghavan2024L4DC}.  For further discussion on the stopping criterion, we refer to \cite{Raghavan2024L4DC}.
\subsection{Estimation at Agents}\label{Estimation at Agents}
At iteration $n$, given the data point $(x^{i}_{n};y^{i}_{n})$ and the downloaded function from iteration $n-1$, $\bar{f}^{i}_{n-1}$, the estimation problem for Agent $i$ in knowledge space $H^{i}$ is,
\begin{align*}
(P1)^{i}_{n}:\; \underset{f^{i}_{n} \in H^{i}} \min C^{i}(f^{i}_n), \;  C(f^{i}_n) = (y^{i}_{n} - f^{i}_{n}(x^{i}_{n}))^{2} +  \\ 
\varrho^{i}_{n}|| f^{i}_{n} - \bar{f}^{i}_{n-1} ||^{2}_{H^i}. 
\end{align*}
The above optimization is a trade off between error (difference between estimated output and true output) and complexity of the difference between the estimates at consecutive iterations controlled by the $\rho^{i}_{n}$ parameter. When  $\rho^{i}_{n}$ is small the first term gets precedence and when  $\rho^{i}_{n}$ is large the latter gets precedence. 

Let $\mathbf{K^{i}} = (K^{i}(\bar{x}^{i}_{j}, \bar{x}^{i}_{k}))_{jk}= (\langle K^{i}(\cdot, \bar{x}^i_k), K^i(\cdot, \bar{x}^i_{j})  \rangle_{H^i})$. Since the inner product is symmetric, $\mathbf{K^{i}} \in \mathbb{R}^{m \times m}$ is symmetric matrix. Let $\mathbf{\bar{K}^{i}}: \mathcal{X} \to \mathbb{R}^m$ be defined as $\mathbf{\bar{K}^{i}}(\cdot) = [K^{i}(\cdot, \bar{x}^{i}_{1}); \ldots; K^{i}(\cdot, \bar{x}^{i}_{m})]$. For any data point, $\mathbf{\bar{K}^{i}}(x^{i}_n) = [K^{i}(x^{i}_{n}, \bar{x}^{i}_{1}); \ldots; K^{i}(x^{i}_{n}, \bar{x}^{i}_{m})]$ is a column vector in $\mathbb{R}^{m}$. Given vector $\boldsymbol{\alpha^i} = [\alpha^{i}_{1}; \ldots; \alpha^{i}_{m}] \in \mathbb{R}^{m}$, we use notation $f^{i} = \boldsymbol{\alpha^{i^T}} \mathbf{\bar{K}^{i}}(\cdot)$ for the function $f^{i} =\sum^{m}_{j=1} \alpha^{i}_{j}K^{i}(\cdot, \bar{x}^i_j) \in H^{i}.$
\begin{proposition}\label{Proposition 2}
Let $ \boldsymbol{\alpha^{i,*}_{n}} =  \Big( \varrho^{i}_{n}\mathbf{K^{i}}+ \mathbf{\bar{K}^{i}}(x^{i}_n) \mathbf{\bar{K}^{i^T}}(x^{i}_n) \Big)^{-1} \\ \Big( \mathbf{\bar{K}^{i}}(x^{i}_n) y^{i}_{n} +  \varrho^{i}_{n}\mathbf{K^{i}} \boldsymbol{\bar{\alpha}^{i}_{n-1}}\Big)$. Then, $f^{i,*}_{n} = \boldsymbol{\alpha^{{i,*}^T}_{n}} \mathbf{\bar{K}^{i}}(\cdot)$ solves the optimization problem in problem $(P1)^{i}_n$. 
\end{proposition}
For the proof, we refer to \cite{Raghavan2024L4DC}, \cite{raghavan2024FunctionFusion}. 
\subsection{Fusion at Fusion Center}\label{Fusion at Fusion Center}
Given $f^{1}_{n}$ and $f^{2}_{n}$, the first goal of the fusion center is to estimate the data points $\{(\hat{x}^{i}_{n,j};\hat{y}^{i}_{n,j})\}^{m}_{j=1}$ which under optimal estimation would result in the functions $f^{1}_{n}$ and $f^{2}_{n}$ being estimated. The same can be formulated as a feasibility problem:
\begin{align*}
(P2)^{i}_{n}: \; \underset{\{(\hat{x}^{i}_{n,j},\hat{y}^{i}_{n,j})\}^{m}_{j=1}} \min \bar{c}^{i}\;\; \text{s.t} \;\;  f^{i}_n = \underset{g^{i}_{n} \in H^{i}} \argmin \sum^{m}_{j=1} (\hat{y}^{i}_{n,j} - \\
g^{i}_{n}(\hat{x}^{i}_{n,j}))^2 + \varrho_{n}||g^{i}_{n}||^{2}
\end{align*}
Given $\{(\hat{x}^{1}_{n,j};\hat{y}^{1}_{n,j})\}^{m}_{j=1} \cup \{(\hat{x}^{2}_{n,j};\hat{y}^{2}_{n,j})\}^{m}_{j=1} $, the fusion problem defined as,
\begin{align*}
(P3)_{n}:\; \underset{f_{n} \in H} \min \; C(f_{n}), \;  C(f_{n}) = \sum_{i=1,2} \sum^{m}_{j=1} (\hat{y}^{i}_{n,j} - \\
f_{n}(\hat{x}^{i}_{n,j}))^2 + \varrho_{n} ||f_{n} ||^{2}_{H}, 
\vspace{-0.15cm}
\end{align*}
is a least squares regression problem. If $f_{n}$ solves the optimization problem, then the  downloaded function at Agent $i$, $\bar{f}^{i}_{n}$, is $\sqrt{\bar{L}^{i}} \circ \Pi_{\mathcal{N}\big(\sqrt{\bar{L}^i}\big)^{\perp}} (f_{n})$ (from Lemma\ref{Lemma 30}).

Proposition 2 in \cite{Raghavan2024L4DC} states that there exists two sets of $m$ (dimension of $H$ is $m$) elements each, $\{\bar{x}^{i}_{j}\}^{m}_{j=1}$, $i=1,2$ such that (i) $\{\bar{x}^{1}_{j}\}^{m}_{j=1} \cap \{\bar{x}^{2}_{j}\}^{m}_{j=1} = \emptyset$; (ii) $\{K(\cdot,\bar{x}^{1}_{j})\}^{m}_{j=1}$  and $\{K(\cdot,\bar{x}^{2}_{j})\}^{m}_{j=1}$, each form a basis for $H$. Thus, for any function $f \in H$, $\exists! \{\alpha^{i}_{j}\}^{m}_{j=1}$ such that $f(\cdot) = \sum^{m}_{j=1} \alpha^{1}_{j}K(\cdot,\bar{x}^{1}_{j}) = \sum^{m}_{j=1}\alpha^{2}_{j}K(\cdot,\bar{x}^{2}_{j})$. Let $\mathbf{K} =  (K(\bar{x}^{p}_{j}, \bar{x}^{q}_{k}))_{jk}= (\langle K(\cdot, \bar{x}^q_k), K(\cdot, \\ \bar{x}^p_{j}) \rangle_{H})_{j,k}, p,q = 1,2, j,k, =1 \ldots m$ be a symmetric matrix in $\mathbb{R}^{2m \times 2m}$ and $\mathbf{\bar{K}}: \mathcal{X} \to \mathbb{R}^{2m}$ be defined as $\mathbf{\bar{K}}(\cdot) = \hspace{-2pt} \Big[K(\cdot, \bar{x}^{1}_{1}); \ldots;  K(\cdot, \bar{x}^{1}_{m});  K(\cdot, \bar{x}^{2}_{1});\ldots ;   K(\cdot, \bar{x}^{2}_{m})\Big]$.   Let $\mathbf{\hat{Y}}^i_n = \Big[\hat{y}^{i}_{n,1}; \ldots; \hat{y}^{i}_{n,m}\Big] \in \mathbb{R}^m$ and $\mathbf{\hat{Y}}_{n} =  \Big[\mathbf{\hat{Y}}^1_{n} ;  \mathbf{\hat{Y}}^2_{n}\Big] \in \mathbb{R}^{2m}$. Let $\mathbf{\check{K}^{i}}: \mathcal{X} \to \mathbb{R}^{m}$ be defined as $\mathbf{\check{K}^{i}}(\cdot) = \Big[K(\cdot, \bar{x}^{i}_{1}); \ldots; K(\cdot, \bar{x}^{i}_{m})\Big]$. A function $f^{i}$ which belongs to $H^{i}$ and $H$ can be expressed as $\boldsymbol{\alpha^{i^T}} \mathbf{\bar{K}^{i}}(\cdot)$ and $\boldsymbol{\hat{\alpha}^{i^T}}\mathbf{\check{K}^{i}}(\cdot)$. Indeed, since $\{K(\cdot, \bar{x}^i_j)\}^{m}_{j=1}$ is basis for $H$, $\exists! \{M^{i}_{kj}\}^{m}_{k=1}$ such that, $ K^{i}(\cdot, \bar{x}^{i}_j) = \hspace{-3pt} \sum^{m}_{k=1} \hspace{-2pt} M^{i}_{kj}K(\cdot, \bar{x}^i_k), \\ \forall j$. Thus, 
\begin{align*}
\hspace{-0.7cm}\sum^{m}_{j=1}\alpha^{i}_{j}K^{i}(\cdot, \bar{x}^{i}_j) = \sum^{m}_{k=1}\sum^{m}_{j=1}M^{i}_{kj}\alpha^{i}_{j}K(\cdot, \bar{x}^i_k) = \sum^{m}_{k=1}\hat{\alpha}^i_{k}K(\cdot, \bar{x}^i_k),
\vspace{-1.5cm}
\end{align*}
which implies that $\boldsymbol{\hat{\alpha}^{i}} = M^{i}\boldsymbol{\alpha^{i}}$ where $M^{i} = (M^i_{kj})_{k,j} \in \mathbb{R}^{m \times m}$.
\begin{proposition}\label{Proposition 3}
$\hat{x}^{i}_{n,j} = \bar{x}^{i}_{j}\; \forall n \in \mathbb{N}, i=1,2, j= 1, \ldots, m$ and $\mathbf{\hat{Y}}^{i}_{n} = \Big(\mathbf{K^{i}} + \varrho^{i}_{n}\mathbb{I}_{m}\Big) \boldsymbol{\alpha^{i}_{n}}$ solve the problem $(P2)^{i}_{n}$, where $f^{i}_{n} = \boldsymbol{\alpha^{i^T}_{n}}\mathbf{\bar{K}^{i}}(\cdot)$. $f^*_{n} = \boldsymbol{\alpha^{T}_{n}}\mathbf{\bar{K}}(\cdot)$, where $\boldsymbol{\alpha_{n}} = \Big(\mathbf{K^{T}}\mathbf{K} + \varrho_{n}\mathbf{K}\Big)^{-1}\Big(\mathbf{K^{T}}\mathbf{\hat{Y}}_{n}\Big)$ solves $(P3)_{n}$.
\end{proposition}
For the proof, we refer to \cite{Raghavan2024L4DC}, \cite{raghavan2024FunctionFusion}.
\begin{figure*}
\begin{center}
\begin{align}
&\hspace{-0.8cm} || \bar{T}^{i}(\varrho^i_{n})|| = \underset{(f; \psi^{i}_{(x,y)}) \neq \theta} \sup  \frac{\Big(M^{i}(x, \varrho^i_n)\big(\varrho^i_n K^{i}\boldsymbol{\alpha} + \mathbf{\bar{K}^{i}}(x)y\big)\Big)^{T} K^{i} \Big(M^{i}(x, \varrho^i_n)\big(\varrho^i_n K^{i}\boldsymbol{\alpha} + \mathbf{\bar{K}^{i}}(x)y\big)\Big)}{|| f ||^{2} + || \psi^{i}_{(x,y)} ||^{2}} = \underset{(f; \psi^{i}_{(x,y)}) \neq \theta} \sup  \frac{\phi^{i}_{n}(\boldsymbol{\alpha};x;y)}{|| f ||^{2} + || \psi^{i}_{(x,y)} ||^{2}} \label{Equation 1} \\ 
&\hspace{0cm} = \underset{(\boldsymbol{\alpha} ; x ; y) \in E^i} \sup  \phi^{i}_{n}(\boldsymbol{\alpha};x;y),  \; \phi^{i}_{n}(\boldsymbol{\alpha};x;y) = \underbrace{\Big(\frac{\mathbf{\bar{K}^{i}}(x) y}{\varrho^{i}_{n}} + \mathbf{K^{i}} \boldsymbol{\alpha}\Big)^{T}}_{\phi^{i^T}_{n,1}} \underbrace{\Big(\mathbf{K^{i}} +\frac{\mathbf{\bar{K}^{i}}(x) \mathbf{\bar{K}^{i^T}}(x) }{\varrho^{i}_{n}}\Big)^{-1}}_{\phi^{i}_{n,2}} \mathbf{K^{i}} \Big(\mathbf{K^{i}}+ \frac{\mathbf{\bar{K}^{i}}(x) \mathbf{\bar{K}^{i^T}}(x) }{\varrho^{i}_{n}}\Big)^{-1} \Big( \frac{\mathbf{\bar{K}^{i}}(x) y}{\varrho^{i}_{n}} +  \mathbf{K^{i}} \boldsymbol{\alpha}\Big).  \label{Equation 2} \\
& \hspace{-0.8cm} ||\bar{T}(\cdot )||^2=  \hspace{-10pt} \underset{\substack{(f^{1};\psi^{1}_{(x^{1};y^{1})}; \\ f^{2};\psi^{2}_{(x^{2};y^{2})}) \neq \theta}}  \sup \hspace{-3pt} \frac{\sum_{i= 1, 2}|| \hat{L}^{i} \circ \bar{T}^{i}(\varrho^{i}_{n})[f^{i};\psi^{i}_{(x^{i};y^{i})}  ] ||^2}{ || f^{1} ||^{2}_{H^{1}}  + || \psi^{1}_{(x^{1};y^{1})} ||^{2}_{H^1}  +  || f^{2} ||^{2}_{H^{2}}  +  || \psi^{2}_{(x^{2};y^2)} ||^{2}_{H^2}} 
\overset{(a)}{\leq} \hspace{-5 pt} \underset{\substack{(f^{1};\psi^{1}_{(x^{1};y^{1})}; \\ f^{2};\psi^{2}_{(x^{2};y^{2})}) \neq \theta}} \sup \hspace{-3pt}  \frac{\sum_{i= 1, 2}|| \hat{L}^{i} ||\; || \bar{T}^{i}(\varrho^{i}_{n})|| \; (||f^{i}||^2 + ||\psi^{i}_{(x^{i};y^{i})}||^2)}{ || f^{1} ||^{2}_{H^{1}}  +  || \psi^{1}_{(x^{1};y^{1})} ||^{2}_{H^1} +   || f^{2} ||^{2}_{H^{2}}   +   || \psi^{2}_{(x^{2};y^2)} ||^{2}_{H^2}} \hspace{-4pt}\label{Equation 3}\\
&\hspace{-0.8cm} \underset{l \to \infty} \lim \hspace{5pt}|| \bar{T}(\ldots)||^{2} \leq \underset{k \to \infty} \lim \;   \underset{\bar{E}} \sup \;  \sum_{i= 1, 2} || \hat{L}^{i} ||\; || \bar{T}^{i}(\varrho^{i}_{n_{k_l}})|| \; (||f^{i}||^2 + ||\psi^{i}_{(x^{i};y^{i})}||^2) \overset{(b)}{ = }  \underset{\bar{E}} \sup  \underset{k \to \infty} \lim \; \sum_{i= 1, 2} || \hat{L}^{i} ||\; || \bar{T}^{i}(\varrho^{i}_{n_{k_l}})|| \; (\ldots) = \sum_{i= 1, 2} || \hat{L}^{i} ||\;(\ldots) \leq 1  , \label{Equation 4}
\end{align}
\end{center}
\vspace{-0.9cm}
\end{figure*}
\section{Abstraction of ColEst2L Algorithm }\label{Section 3}
In this section, we decompose the algorithm presented in subsection \ref{Algorithm Description} into sub-blocks and find an operator which abstracts each sub-block. We prove the existence of an estimation operator for any iteration $n$  and until iteration $n$ of the algorithm constructively, subsection \ref{Subsection 3.5}, thus solving Problem \ref{Problem 5}. We define the estimation operator for iteration $n$ by composing:
\begin{itemize}
    \item an operator from $H^{1} \times H^{2} \times H^{1} \times H^{2} $ to $H \times H $ which abstracts the estimation at the agents followed the upload to the fusion center: subsections \ref{Subsection 3.1}, \ref{Subsection 3.2};
    \item an operator which captures the fusion problem: subsection \ref{Subsection 3.3} ;
    \item an operator which represents the download from $H$ to $H^{1} \times H^{2}$: subsection \ref{Subsection 3.4}. 
\end{itemize}
\subsection{Operator for Estimation at The Agents}\label{Subsection 3.1}
Given data point $(x,y)$ and current knowledge as $f$, we could define the estimation operator at Agent $i$, $T^{i}: H^{i} \to H^{i} $ as $T^{i}(x;y; \varrho)[f] = \underset{g \in H^{i}}\argmin \; (y -g(x))^2 + \varrho|| f - g ||^2_{H^{i}}$.  Let $M: \mathcal{X} \times \mathbb{R} \to \mathbb{R}^{m \times m}$ be defined as $M^i(x; \varrho) =  \Big( \varrho\mathbf{K^{i}}+ \mathbf{\bar{K}^{i}}(x) \mathbf{\bar{K}^{i^T}}(x) \Big)^{-1}$. Then, from Proposition \ref{Proposition 2} the estimation operator can be defined as:
\begin{align*}
   \hspace{-0.7cm} T^{i}(x;y;\varrho)[f] =\Big(M^{i}(x, \varrho)\big(\varrho K^{i}\boldsymbol{\alpha} + (\mathbf{\bar{K}^{i}}(x)y)^T\big)\Big)^{T}\mathbf{\bar{K}^{i}} (\cdot), 
\end{align*}
where $ f = \boldsymbol{\alpha^T}\mathbf{\bar{K}^{i}}(\cdot)$.  However, $T^{i}(x;y;\varrho)[\cdot]$ is an affine mapping in $f$ . In order to avoid the complications associated with the definition of a norm for a nonlinear operator and its properties specifically with respect to composition of operators, we define the estimation operator as follows:
\begin{definition}\label{Definition 7}
Given data point, $(x;y)$, we define $\psi^{i}_{(x;y)}$ as $ \psi^{i}_{(x;y)} = (\mathbf{\bar{K}^{i}}(x)y)^T\mathbf{\bar{K}^{i}}(\cdot)$. The estimation operator at Agent $i$, $\bar{T}^{i}: H^{i} \times H^{i} \to H^{i}$ is defined as $\bar{T}^{i}(\varrho)[ f;\psi^{i}_{(x,y)} ] =   \Big(M^{i}(x; \varrho)\big(\varrho K^{i}\boldsymbol{\alpha} + \mathbf{\bar{K}^{i}}(x)y\big)\Big)^{T}\mathbf{\bar{K}^{i}} (\cdot)$, where $f= \boldsymbol{\alpha^T}\mathbf{\bar{K}^{i}} (\cdot)$. $\bar{T}^{i}(\varrho)[\cdot]$ is linear and bounded.
\end{definition}
\begin{lemma}\label{Lemma 8}
Let $\{\varrho^i_{n}\}$ be sequence of real numbers such that $\varrho^i_{n} \to \infty$. Then, there is a subsequence $\{\varrho^i_{n_k}\}$ of $\{\varrho^i_{n}\}$ such that $\underset{k \to \infty} \lim || \bar{T}^{i}(\varrho^i_{n_k})|| = 1$.
\end{lemma}
\begin{proof}
Let $f =  \boldsymbol{\alpha^{T}}\mathbf{\bar{K}^{i}} (\cdot)$ and $\psi^{i}_{(x;y)}$ be as defined above. The set $\{f \in H^{i},  x \in \mathcal{X}, y \in \mathbb{R}\;:\: || f ||^{2} + || \psi^{i}_{(x,y)} ||^{2} =1 \}$ is isomorphic to the set $E^{i} =\{\boldsymbol{\alpha} \in \mathbb{R}^{m}, x \in \mathcal{X}, y \in \mathbb{R}: \boldsymbol{\alpha^{T}}\mathbf{K^{i}}\boldsymbol{\alpha^{T}} + y^2 \mathbf{\bar{K}^{i^T}}(x)\mathbf{K^{i}} \mathbf{\bar{K}^{i}}(x)= 1 \}$. $E^{i}$ is a compact subset of $\mathbb{R}^m \times \mathbb{R}^I \times \mathbb{R}$. For any $n$, the norm of the estimation operator is expanded in Equation \ref{Equation 1} and the numerator is expanded in Equation \ref{Equation 2}. From Proposition \ref{Proposition 34}, the product of uniformly bounded equicontinuous sequence of functions is a uniformly bounded equicontinuous sequence of functions. By this argument and Propositions \ref{Proposition 31}, and, \ref{Proposition 33}, the sequences,  $\{\mathbf{K^{i}}\phi^{i}_{n,2}(\cdot)\phi^{i}_{n,1}(\cdot)\}$, $\{\phi^{i^T}_{n,2}(\cdot)\mathbf{K^{i}}\phi^{i}_{n,2}(\cdot)\phi^{i}_{n,1}(\cdot)\}$, $\{ \phi^{i^T}_{n,1}(\cdot) \phi^{i^T}_{n,2}(\cdot) \mathbf{K^{i}}\phi^{i}_{n,2}(\cdot)\phi^{i}_{n,1}(\cdot)\}$ are equicontinuous on $E^{i}$. Thus, $\{\phi^{i}_{n}(\cdot)\}$ is sequence of equicontinuous functions on $E^{i}$ which converges pointwise to $\boldsymbol{\alpha}^{T}\mathbf{K^{i}}\boldsymbol{\alpha}\; \text{on}\; E^{i}$. By the \textit{Arzel\`a Ascoli Theorem} (refer chpater 10, \cite{royden2010real}), there exists a subsequence $\{\phi^{i}_{n_k}(\cdot)\}$ that converges uniformly to $\boldsymbol{\alpha}^{T}\mathbf{K^{i}}\boldsymbol{\alpha}\; \text{on}\; E^{i}$. Hence, the limit and the supremum can be swapped as below, 
\begin{align*}
\underset{k \to \infty}\lim || \bar{T}^{i}(&\varrho^{i}_{n_k})||^2 = \underset{k \to \infty}\lim \; \underset{ E^{i}} \sup \; \phi^{i}_{n_k}(\boldsymbol{\alpha};x;y) = \\
&\underset{E^{i}} \sup \; \underset{k \to \infty}\lim \; \phi^{i}_{n_k}(\boldsymbol{\alpha};x;y) = \underset{E^{i}} \sup \; \boldsymbol{\alpha}^{T}\mathbf{K^{i}}\boldsymbol{\alpha} = 1,
\end{align*}
proving the statement of the Lemma. 
\end{proof}
\subsection{The Multi-Agent Estimation and Upload Operator}\label{Subsection 3.2}
In this subsection, we define an operator which accounts for estimation at the agents and the upload operation comprehensively. This operator is obtained through the composition of the estimation operator and the upload (or communication) operator, Corollary \ref{Corollary 28}, at the agents.
\begin{definition}\label{Definition 9}
The multi-agent estimation and upload operator, $\bar{T}$, is defined as $\bar{T}: H^{1} \times H^{2} \times H^{1} \times H^{2}  \to H \times H$ defined as
\begin{align*}
&\bar{T}(\varrho^{1};\varrho^{2})[f^{1}; f^{2};\psi^{1}_{(x^{1};y^{1})}; \psi^{2}_{(x^{2};y^{2})}] =  \\
&\Big[\hat{L}^{1} \circ \bar{T}^{1}(\varrho^{1}) [  f^{1};  \psi^{1}_{(x^{1};y^{1})}]; \hat{L}^{2} \circ \bar{T}^{2}(\varrho^{2})[f^{2}; \psi^{2}_{(x^{2};y^{2})}]\Big].
\end{align*}
\end{definition}
\begin{lemma}\label{Lemma 10}
Let $\{\varrho^{1}_{n}\}$ and $\{\varrho^{2}_{n}\}$ be sequences which diverge to $\infty$. Then, there exists  a subsequence of $\{\bar{T}(\varrho^{1}_{n}; \varrho^2_n)\}$, $\{\bar{T}(\varrho^{1}_{n_k}; \varrho^2_{n_K})\}$ such that $\underset{k \to \infty} \lim || \bar{T}(\varrho^{1}_{n_k}; \varrho^2_{n_K})|| \leq  1 $  
\end{lemma}
\begin{proof}
The square of the norm of the multi-agent estimation operator, $|| \bar{T}(\varrho^{1}_{n};\varrho^{2}_{n})||^{2}$, is expanded in Equation \ref{Equation 3}. Inequality $(a)$ follows from the property of operator norms, $|| T \circ L || \leq || T || \; || L || $. Consider the two subsequences from Lemma \ref{Lemma 8}. From the subsequence $\{n_k\}$ for Agent 1, we can extract a subsequence $\{n_{k_l}\}$ for Agent 2  such that  $\underset{l \to \infty} \lim || \bar{T}^{2}(\varrho^i_{n_{k_l}})|| = 1$. Hence, there exits a common subsequence such that $\underset{l \to \infty} \lim || \bar{T}^{i}(\varrho^i_{n_{k_l}})|| = 1, i=1,2$.  Consider the set $\bar{E} = \{ (f^{1};\psi^{1}_{(x^{1};y^{1})}; f^{2};\psi^{2}_{(x^{2};y^{2})}) \in H^{1} \times H^{2}\times H^{1} \times H^{2} : \sum_{i=1,2} || f^{i} ||^{2}_{H^{i}}  + || \psi^{i}_{(x^{i};y^{i})} ||^{2}_{H^i} =1\}$. This is a compact subset of its domain. Moreover since $||\hat{L}^i || \leq 1$ (Corollary \ref{Corollary 28}) and  $|| \bar{T}^{i}(\varrho^i_{n_{k_l}})|| $ is bounded, the sequence
\begin{align*}
    \hspace{-0.7cm}\tilde{\phi}_{l} = || \hat{L}^{1} ||\; &|| \bar{T}^{1}(\varrho^{1}_{n_{k_l}})|| \; (||f^{1}||^2 + ||\psi^{1}_{(x^{1};y^{1})}||^2) + \\ &|| \hat{L}^{2} ||\; || \bar{T}^{2}(\varrho^{2}_{n_{k_l}})|| \; (||f^{2}||^2 + ||\psi^{2}_{(x^{2};y^{2})}||^2),
\end{align*}
is equicontinuous on $\bar{E}$. Hence by the \textit{Arzel\`a Ascoli Theorem}, there is a subsequence of $\{\tilde{\phi}_{l}\}$ which converges uniformly to $\sum_{i=1,2} || \hat{L}^{i} || (||f^{i}||^2 + ||\psi^{i}_{(x^{i};y^{i})}||^2) $. By swapping the limit and the supremum, as done in equality $(b)$, an upper bound on the limit of $|| \bar{T}(\varrho^{1}_{n_{k_l}};\varrho^{2}_{n_{k_l}})||^{2}$ as $l$ tends to infinity has been derived in Equation \ref{Equation 4}.  For brevity, we continue to refer to the sequences $\{\bar{T}^{i}(\varrho^{i}_{n_{k_l}}\}$ in Equation  \ref{Equation 4} even though the result only holds for a subsequence of it.  
\end{proof}
We note that mapping from $(x;y) \mapsto \psi^{i}_{(x;y)}$ is a nonlinear map however bounded under suitable assumptions on $K$. This mapping has not been incorporated into the definition of $\bar{T}$ to avoid complications associated with nonlinearity. 
\subsection{Operator for Fusion at The Fusion Center}\label{Subsection 3.3}
Given the functions $f^{1} \in \{H \cap H^{1}\}$ and $f^{2} \in \{H\cap H^{2}\}$ uploaded to fusion center, the closed form expression for the fused function from Proposition \ref{Proposition 3} is utilized to define the fusion operator as follows:
\begin{definition}\label{Definition 11}
The fusion operator at the fusion center is $T : \{H\cap H^{1}\} \times \{H \cap H^{2}\}  \to H$ defined as $T(\varrho)[f^{1},f^{2}] = \boldsymbol{\alpha^T}\mathbf{\bar{K}}$ where $\boldsymbol{\alpha}$ is 
\begin{align*}
\boldsymbol{\alpha} = \Big(\mathbf{K^{T}}\mathbf{K} + \varrho\mathbf{K}\Big)^{-1}
\Bigg[\mathbf{K^{T}} 
\begin{bmatrix}
\Big(\mathbf{K^{1}} + \varrho \mathbb{I}_{m}\Big) \boldsymbol{\alpha^{1}} \\
\Big(\mathbf{K^{2}} + \varrho\mathbb{I}_{m} \Big) \boldsymbol{\alpha^{2}}
\end{bmatrix} \Bigg], 
\end{align*}
where $f^{i} = \boldsymbol{\alpha^{i}}\mathbf{\bar{K}^{i}}(\cdot) = \boldsymbol{\hat{\alpha}^{i^T}}\mathbf{\check{K}^{i}}(\cdot)= \boldsymbol{\alpha^{i^T}}M^{i^T}\mathbf{\check{K}^{i}}(\cdot)$. 
\end{definition}
We note that, in  the definition of $T$ the domain is restricted to product of the subspaces of $H$ which contain $H^{1}$ and $H^{2}$. This is because the functions received by the fusion center belongs to these subspaces. Further, the closed form expression for the fused function holds if and only if the input functions belong to the corresponding subspaces. This translates into the definition of the norm as well as in the proof below.
\begin{lemma}\label{Lemma 12}
Let 
\begin{align*}
\hspace{2cm}
\mathbf{K} =\begin{bmatrix} 
&\mathbf{\tilde{K}^{1}} &\mathbf{K^{12}} \\ 
&\mathbf{K^{{12}^T}} &\mathbf{\tilde{K}^{2}} 
\end{bmatrix} ,
\end{align*}
be such that $\max(\lambda_{\max}(\mathbf{\tilde{K}^{1}}), \lambda_{\max}(\mathbf{\tilde{K}^{2}})) \leq 1$ and $\varrho_{n}$ be monotone sequence of positive real numbers diverging to $\infty$. Then, $\underset{ n \to \infty} \lim || T(\varrho_n) || \leq 1$. 
\end{lemma}
\begin{proof}
The set $\{(f^1;f^2) \in \{H \cap H^{1}\} \times \{ H \cap H^{2}\} : || f^{1} ||^2_{H} + || f^{2} ||^2_{H} =1 \}$ is isomorphic to the set $E$ defined in Equation \ref{Equation 5}. The norm of the estimation operator at the fusion center at iteration $n$ is
\begin{align*}
    || T(\varrho_{n}) ||^2 = \underset{(\boldsymbol{\alpha^{1}};\boldsymbol{\alpha^{2}}) \in E} \sup \;   \phi_{n}(\boldsymbol{\alpha^{1}};\boldsymbol{\alpha^{2}}) 
\end{align*}
where $\phi_{n}$ is defined in Equation \ref{Equation 6}. From Proposition \ref{Proposition 35}, we conclude that $\{\mathbf{K}\phi_{n,2}(\cdot)\phi_{n,1}(\cdot)\}$, $\{\phi^{T}_{n,2}(\cdot)\mathbf{K}\phi_{n,2}(\cdot)\phi_{n,1}(\cdot)\}$ and $\{\phi^{T}_{n,1}(\cdot)\phi^{T}_{n,2}(\cdot)\mathbf{K}\phi_{n,2}(\cdot)  \phi_{n,1}(\cdot)\}$ are uniformly bounded, uniformly converging sequences of continuous functions. Thus, the sequence of continuous functions $\{\phi_{n}(\cdot, \cdot)\}$ converges uniformly to  the continuous function $ \phi(\boldsymbol{\alpha^{1}};\boldsymbol{\alpha^{2}}) =\boldsymbol{\alpha^T}\mathbf{K} \boldsymbol{\alpha}$ where $\boldsymbol{\alpha}= (\boldsymbol{\alpha^{1}};\boldsymbol{\alpha^{2}})$. Thus, swapping limits and supremum, we get
\begin{align*}
\underset{n \to \infty} \lim || T(\varrho_{n}) ||^2 =& \underset{n \to \infty} \lim \underset{(\boldsymbol{\alpha^{1}};\boldsymbol{\alpha^{2}}) \in E} \sup \;   \phi_{n}(\boldsymbol{\alpha^{1}};\boldsymbol{\alpha^{2}}) \\  =&
\underset{(\boldsymbol{\alpha^{1}};\boldsymbol{\alpha^{2}}) \in E} \sup  \underset{n \to \infty} \lim  \phi_{n}(\boldsymbol{\alpha^{1}};\boldsymbol{\alpha^{2}}) \\
= & \underset{(\boldsymbol{\alpha^{1}};\boldsymbol{\alpha^{2}}) \in E} \sup \; \boldsymbol{\alpha^T}\mathbf{K} \boldsymbol{\alpha} \leq \lambda_{\text{max}}(\mathbf{K}).
\end{align*}
$\lambda_{\text{max}}(\mathbf{K}) \leq 1$ would be sufficient for the statement in the  current proposition to hold. Consider the block decomposition of $\mathbf{K}$ in Equation \ref{Equation 7} obtained by invoking \textit{Schur's complement}. Utilizing this decomposition of $\mathbf{K}$ we obtain 
\begin{align*}
&\lambda_{\text{max}}(\mathbf{K}) =\underset{x \neq 0} \sup \;  \frac{x^{T}\mathbf{K}x}{x^Tx} = \underset{x \neq 0} \sup \;  \frac{x^{T}P^T D Px}{x^Tx}  \\
&\overset{(a)}{=} \underset{y \neq 0} \sup \;  \frac{y^{T} D y}{y^Ty} = \lambda_{\text{max}}(D) \overset{(b)}{\leq} \lambda_{\text{max}}(\mathbf{\tilde{K}^{1}}),\lambda_{\text{max}}(\mathbf{\tilde{K}^{2}}).
\end{align*}
Equality $(a)$ is true because $P$ is a full rank matrix ($\mathcal{R}(P) = \mathbb{R}^{2m}$) with all eigenvalues equal to $1$. Note that $\mathbf{K^{12}}\mathbf{\tilde{K}^{2^{-1}}}\mathbf{K^{{12}^T}}$ is positive definite. Thus, 
\begin{align*}
    \lambda_{\max}(\mathbf{\tilde{K}^{1}}  - \mathbf{K^{12}}\mathbf{\tilde{K}^{2^{-1}}}\mathbf{K^{{12}^T}}) \leq \lambda_{\max}(\mathbf{\tilde{K}^{1}})
\end{align*}
The set of eigenvalues of $D$ is the union of the set of eigenvalues of $\mathbf{\tilde{K}^{1}}  - \mathbf{K^{12}}\mathbf{\tilde{K}^{2^{-1}}}\mathbf{K^{{12}^T}}$ and $\mathbf{\tilde{K}^{2}}$ which implies that,
\begin{align*}
    \lambda_{\text{max}}(D) \leq \lambda_{\text{max}}(\mathbf{\tilde{K}^{1}} - \mathbf{K^{12}}\mathbf{\tilde{K}^{2^{-1}}}\mathbf{K^{{12}^T}}),\lambda_{\text{max}}(\mathbf{\tilde{K}^{2}}). 
\end{align*}
Hence, inequality $(b)$ follows. Thus, 
\begin{align*}
    \underset{n \to \infty} \lim || T(\varrho_{n}) ||^2 \leq \max(\lambda_{\text{max}}(\mathbf{\tilde{K}^{1}}),\lambda_{\text{max}}(\mathbf{\tilde{K}^{2}})) \leq 1
\end{align*}
It is not necessary that $\lambda_{\max}(\mathbf{K}) \leq 1 $ or $\max(\lambda_{\max}(\mathbf{\tilde{K}^{1}}), \\ \lambda_{\max}(\mathbf{\tilde{K}^{2}})) \leq 1$. However, by suitably normalizing $\mathbf{K}$ or $\mathbf{\tilde{K}^{1}}, \mathbf{\tilde{K}^{2}}$ matrices for large $n$, we can ensure that $\lambda_{\text{max}}(\mathbf{K}) \leq 1$ and thus guarantee that $|| T(\varrho_{n_{k_l}})||^2 $ converges to a value less than or equal to $1$.
\end{proof}
\begin{figure*}
\begin{align}
&E = \{\boldsymbol{\hat{\alpha}} = (\boldsymbol{\hat{\alpha}^{1}};\boldsymbol{\hat{\alpha}^{2}}) \in \mathbb{R}^{2m}: \exists \boldsymbol{\alpha^{i}} \text{ s.t } \boldsymbol{\hat{\alpha}^{i}} = M^{i}\boldsymbol{\alpha^{i}}, \text{ and }    \boldsymbol{\hat{\alpha}^{1^T}}\mathbf{\breve{K}^{1}} \boldsymbol{\hat{\alpha}^{1}} + \boldsymbol{\hat{\alpha}^{2^T}}\mathbf{\breve{K}^{2}} \boldsymbol{\hat{\alpha}^{2}} = 1 \} \nonumber \\
&\hspace{3.5cm} = \{ \boldsymbol{\alpha} = (\boldsymbol{\alpha^{1}};\boldsymbol{\alpha^{2}}) \in \mathbb{R}^{2m}: \boldsymbol{\alpha^{1^T}}M^{1^T}\mathbf{K^{1}} M^{1}\boldsymbol{\alpha^{1}} + \boldsymbol{\alpha^{2^T}}M^{2^T}\mathbf{K^{2}} M^{2}\boldsymbol{\alpha^{2}}  = 1\} \label{Equation 5} \\
&\phi_{n}(\boldsymbol{\alpha^{1}};\boldsymbol{\alpha^{2}}) = \bigg[\Big[\boldsymbol{\alpha^{1^T}} \Big(\frac{\mathbf{K^{1}}}{\varrho_n} +  \mathbb{I}_{m}\Big), \boldsymbol{\alpha^{2^T}} \Big(\frac{\mathbf{K^{2}}}{\varrho_n} +  \mathbb{I}_{m}\Big) \Big] \mathbf{K} \bigg] \Big(\frac{\mathbf{K^{T}}\mathbf{K}}{\varrho_n} + \mathbf{K}\Big)^{-1} \mathbf{K}
\underbrace{\Big(\frac{\mathbf{K^{T}}\mathbf{K}}{\varrho_n} + \mathbf{K}\Big)^{-1}}_{\phi_{n,2}(\cdot)}\underbrace{\Bigg[\mathbf{K^{T}} 
\begin{bmatrix}
\Big(\frac{\mathbf{K^{1}}}{\varrho_n} +  \mathbb{I}_{m}\Big) \boldsymbol{\alpha^{1}} \\
\Big(\frac{\mathbf{K^{2}}}{\varrho_n} + \mathbb{I}_{m} \Big) \boldsymbol{\alpha^{2}}
\end{bmatrix} \Bigg]}_{\phi_{n,1}(\cdot)}.\label{Equation 6} \\
&\hspace{1cm}\mathbf{K} = 
\begin{bmatrix} 
&\mathbf{\tilde{K}^{1}} &\mathbf{K^{12}} \\ 
&\mathbf{K^{{12}^T}} &\mathbf{\tilde{K}^{2}} 
\end{bmatrix} =
\underbrace{\begin{bmatrix} 
&\mathbb{I}_{m} &\mathbf{K^{12}}\mathbf{\tilde{K}^{2^{-1}}} \\ 
&0 &\mathbb{I}_{m}
\end{bmatrix}}_{P^T}\hspace{-4pt}
\underbrace{\begin{bmatrix}
&\mathbf{\tilde{K}^{1}}  - \mathbf{K^{12}}\mathbf{\tilde{K}^{2^{-1}}}\mathbf{K^{{12}^T}} &0\\
&0 &\mathbf{\tilde{K}^{2}} 
\end{bmatrix}}_{D}\hspace{-4pt}
\underbrace{\begin{bmatrix}
&\mathbb{I}_{m} &\mathbf{K^{12}}\mathbf{\tilde{K}^{2^{-1}}} \\ 
&0 &\mathbb{I}_{m}
\end{bmatrix}^{T}}_{P}\label{Equation 7}\\
&\hspace{-0.8cm}c^{2}_{d}|| \hat{T} ||^{2} = \underset{(f \neq \theta)} \sup \; \frac{ || \sqrt{\bar{L}^{1}} \circ \Pi_{\mathcal{N}\big(\sqrt{\bar{L}^1}\big)^{\perp}}\big(f\big) ||^{2}_{H^1} + || \sqrt{\bar{L}^{2}} \circ \Pi_{\mathcal{N}\big(\sqrt{\bar{L}^2}\big)^{\perp}}\big(f\big) ||^{2}_{H^2}}{ || f ||^{2}_{H}} \overset{(c)}{=}  \underset{(f \neq \theta)} \sup \frac{||\Pi_{\mathcal{N}\big(\sqrt{\bar{L}^1}\big)^{\perp}}\big(f\big) ||^{2}_{H} + || \Pi_{\mathcal{N}\big(\sqrt{\bar{L}^2}\big)^{\perp}}\big(f\big) ||^{2}_{H}}{ || f ||^{2}_{H}} = c^{2}_d.\label{Equation 8}\\
&\hspace{-0.8cm}\bar{\mathbb{T}}_{n}(\{(\varrho^{1}_k;\varrho^{2}_k; \varrho_k)\}^{n}_{k=1})\Big[f^{1}_{0};f^{2}_{0};\{\psi^{1}_{(x^{1}_j;y^{1}_j)}; \psi^{2}_{(x^{2}_j;y^{2}_j)}\}^n_{j=1}\Big] = \nonumber \\ 
&\hspace{5cm} \mathbb{T}_{n}(\varrho^{1}_{n}; \varrho^{2}_{n}; \varrho_n)\Big[\bar{\mathbb{T}}_{n-1}(\{(\varrho^{1}_k;\varrho^{2}_k; \varrho_k)\}^{n-1}_{k=1}\})\Big[\bar{f}^{1}_{n-2};\bar{f}^{2}_{n-2} ;\{\psi^{1}_{(x^{1}_{j};y^{1}_{j})}; \psi^{2}_{(x^{2}_{j};y^{2}_{j})}\}^{n-1}_{j=1}\} \Big];\psi^{1}_{(x^{1}_{n};y^{1}_{n})}; \psi^{2}_{(x^{2}_{n};y^{2}_{n})}\Big]\;, \nonumber\\
&\hspace{-0.8cm}\bar{\mathbb{T}}_{1}(\varrho^{1}_1;\varrho^{2}_1; \varrho_1)\Big[f^{1}_{0};f^{2}_{0};\psi^{1}_{(x^{1}_1;y^{1}_1)}; \psi^{2}_{(x^{2}_1;y^{2}_1)}\Big] =  \mathbb{T}_{1}(\varrho^{1}_1;\varrho^{2}_1; \varrho_1)\Big[f^{1}_{0};f^{2}_{0};\psi^{1}_{(x^{1}_1;y^{1}_1)}; \psi^{2}_{(x^{2}_1;y^{2}_1)}\Big], \label{Equation 9} 
\end{align}
\vspace{-0.9cm}
\end{figure*}
\subsection{The Multi-Agent Download Operator}\label{Subsection 3.4}
Given the fused function in $H$, it is to be downloaded onto $H^{1}$ and $H^2$. The download operator from $H$ to $H^{i}$ is defined in Lemma \ref{Lemma 30}. Since every iteration of algorithm begins with estimates in $H^{1} \times H^2$, the iteration is to conclude with estimates in $H^{1} \times H^{2}$. We define the multi-agent download operator by concatenating the downloads on to $H^1$ and $H^2$ followed by normalization: 
\begin{definition}\label{Definition 13}
The multi-agent download operator, $\hat{T}: H \to H^{1} \times H^{2}$, is defined as 
\begin{align*}
\hat{T}(f) = \frac{1}{c_d} \Big[\sqrt{\bar{L}^{1}} \circ \Pi_{\mathcal{N}\big(\sqrt{\bar{L}^1}\big)^{\perp}}&\big(f\big) ; \\
&\sqrt{\bar{L}^{2}} \circ \Pi_{\mathcal{N}\big(\sqrt{\bar{L}^2}\big)^{\perp}}\big(f\big)  \Big] 
\end{align*}
where $\bar{L}^{i}$ is defined in Lemma \ref{Lemma 29} and 
\begin{align*}
    c_{d} = \underset{(f \neq \theta)} \sup \; 1 + \frac{\langle f, \Pi_{\mathcal{N}\big(\sqrt{\bar{L}^1}\big)^{\perp}}\Pi_{\mathcal{N}\big(\sqrt{\bar{L}^2}\big)^{\perp}}\big(f \big) \rangle_{H}}{ || f||^{2}_H}
\end{align*}
The norm of $\hat{T}$, $||\hat{T} ||$, is equal to 1.
\end{definition}
The unnormalized norm square  of  $\hat{T}$ has been derived in Equation \ref{Equation 8}  where equality $(c)$ is obtained invoking Lemma \ref{Lemma 30}. Thus, the unnormalized norm square  of  $\hat{T}$ is $c^{2}_d$, while the norm of the normalized operator is $1$. Normalization was not needed in the proof of Lemma \ref{Lemma 30}, but is needed here because of the concatenation of the downloads, i.e., change in the download space from $H^{i}$ to $H^{1} \times H^2$. We note that the normalization of the download operator is not needed for all $n$. For $n$ sufficiently large, we normalize the downloaded functions by $c_d$ to so that $|| \hat{T} || = 1$. This done to ensure that the estimation operator for iteration $n$ has norm converging to $1$.

\subsection{Estimation Operator for Iteration $n$}\label{Subsection 3.5}
In Figure \ref{Figure 3},  the operators defined so far and spaces they map across has been captured. Operator $\bar{T}(\varrho^{1}_n, \varrho^2_n)$ takes in $(\bar{f}^{1}_{n-1}; \psi^{1}_{x^1_n, y^1_n}, \bar{f}^{2}_{n-1}; \psi^{2}_{x^2_n, y^2_n})$ and outputs $f^{1}_n, f^{2}_n$, operator $T(\varrho_n)$ takes in $f^{1}_n$ and $f^2_n$ and outputs $f_{n}$, operator $\hat{T}$ takes in $f_{n}$ and outputs $\bar{f}^{1}_n$ and $\bar{f}^{2}_n$. This corresponds to one iteration of the algorithm, specifically iteration $n$. Thus, the estimation operator at iteration $n$ is defined as the composition of the operators, $\bar{T}(\varrho^{1}_n, \varrho^2_n)[\cdot]$, $T(\varrho_n)[\cdot]$, and $\hat{T}(\cdot)$, as below.
\begin{definition}\label{Definition 14}
The estimation operator at iteration $n$ is defined as $ \mathbb{T}_{n}: H^{1} \times H^{2} \times H^{1} \times H^{2}\to H^{1} \times H^{2}$, 
\begin{align*}
 \mathbb{T}_{n}(\varrho^{1}_{n}; \varrho^{2}_{n}; &\varrho_n)[f^{1};f^{2};\psi^{1}_{(x^{1};y^{1})}; \psi^{2}_{(x^{2};y^{2})}] = \\
 & \hspace{-0.4cm}\hat{T} \circ  T(\varrho_{n}) \circ \bar{T}(\varrho^{1}_{n}; \varrho^{2}_{n})[f^{1};f^{2}; \psi^{1}_{(x^{1};y^{1})}; \psi^{2}_{(x^{2};y^{2})}].  
\end{align*}
\end{definition}
Operator $\mathbb{T}_{n}(\cdot)[\cdot]$ has also been captured in Figure \ref{Figure 3}. The estimation operator until iteration $n$ could be defined by composing the estimation operators at iterations $1, \ldots, n$. Note that the domain for these operators is $H^1 \times H^2  \times H^1 \times H^2$ while the range is $H^1 \times H^2$. Hence, direct composition is not possible. The two extra input terms correspond to two new data points which are received at every time step.  Accounting for the same, the operator is defined as follows. 
\begin{definition}\label{Definition 15}
Given $(f^{1}_0;f^{2}_0)$ and a finite string of input data $\{(x^{1}_j; y^{1}_{j})\}^{n}_{j=1}, \{(x^{2}_j;y^{2}_j)\}^n_{j=1}$, the estimation  operator until iteration $n$, $\bar{\mathbb{T}}_{n} : H^{1} \times H^{2} \times (H^{1} \times H^{2})^{n} \to H^{1} \times H^{2}$ is defined recursively in Equation \ref{Equation 9}.
\end{definition}
Since the composition of linear operators is linear, $\{\mathbb{T}_{n}\}$ and $\{\bar{\mathbb{T}}_{n}\}$ are sequences of linear operators. 
\begin{remark}\label{Remark 36}
In Lemmas \ref{Lemma 8}, \ref{Lemma 10}, \ref{Lemma 12} and Definition \ref{Definition 13}, we have insisted that the norm of operators converge to a value less than or equal to $1$. From definition of $\mathbb{T}_{n}$, it follows that $ || \mathbb{T}_{n}  ||  \leq  || \bar{T}(\varrho^{1}_{n}; \varrho^2_{n})|| \; ||T(\varrho_n)|| \; || \hat{T} ||$. Since each sequence on R.H.S  has a subsequence which converges to a value less than or equal 1, there is a subsequence of $\{ \mathbb{T}_{n} \}$ whose norm converges to a value less than or equal 1. Using the same idea, the norm of $\bar{\mathbb{T}}_{n}$ is dependent on $\prod^{n}_{j=1} || \mathbb{T}_{j} ||$. Though, $\{ || \mathbb{T}_{n} || \}$ converges to a value $\leq 1$, it possible that its product diverges. However, we can choose a subsequence such that $\prod^{n}_{j=1} || \mathbb{T}_{j} ||$ is finite as $n \to \infty$, thus "uniformly bounding" $\{\bar{\mathbb{T}}_{n}\}$. This is crucial for the existence of the estimation operator for the algorithm and proof of consistency. Hence, through Lemmas \ref{Lemma 8}, \ref{Lemma 10}, \ref{Lemma 12}, we have in principle solved item (iii) of Problem \ref{Problem 6}. 
\end{remark}
\begin{figure*}
\begin{center}
\includegraphics[scale=0.44]{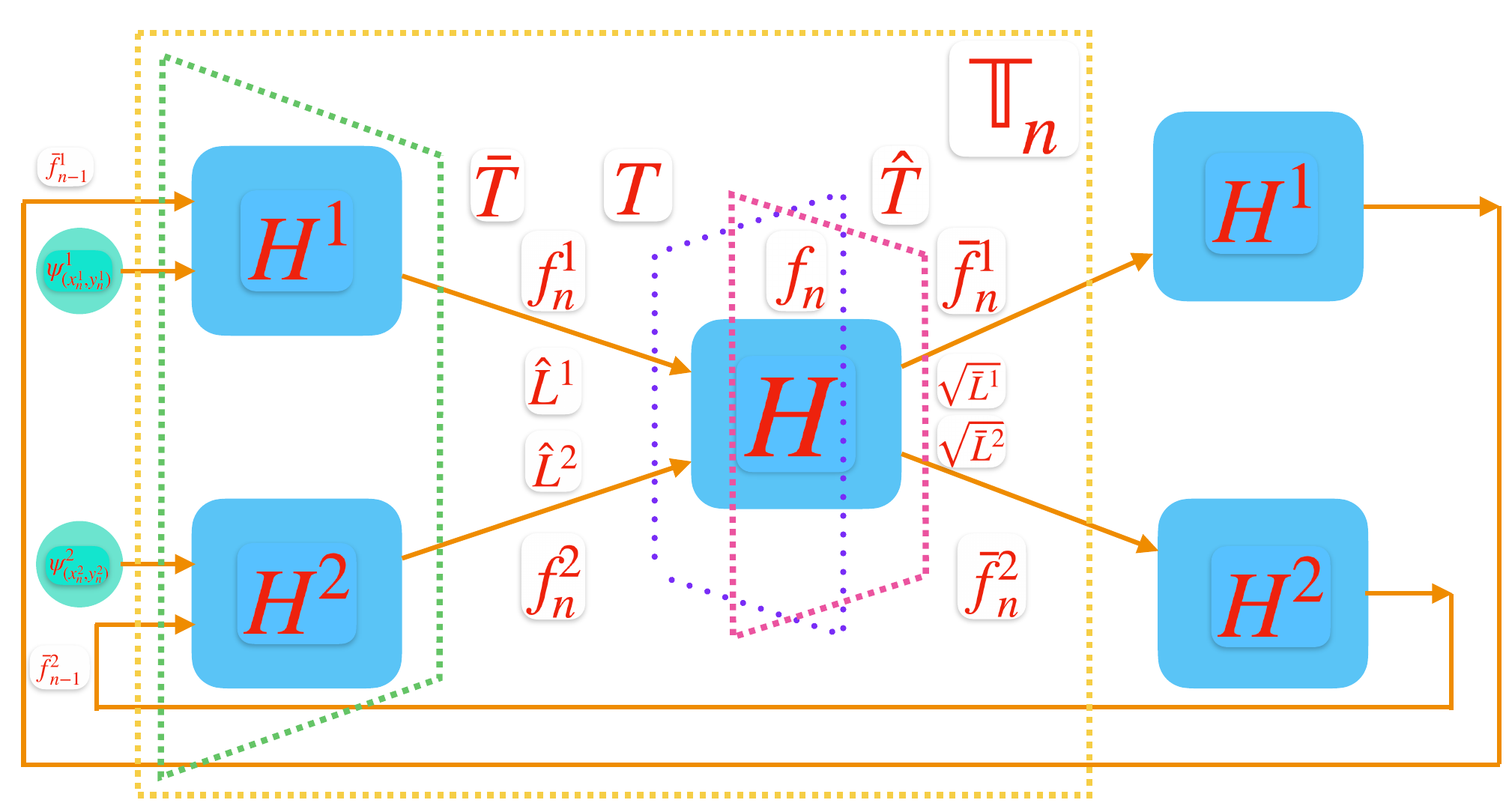}
\caption{The estimation operator at iteration $n$, $\mathbb{T}_{n}$, obtained through composition of operators, $\bar{T}, T, \hat{T}.$}\label{Figure 3}
\end{center}
\vspace{-0.9cm}
\begin{align}
&\hspace{-0.8cm} || \bar{\mathbb{T}}_{n}(\cdot)\Big[f^{1}_{0};f^{2}_{0};\{\psi^{1}_{(x^{1}_1;y^{1}_1)}; \psi^{2}_{(x^{2}_1;y^{2}_1)}\}^n_{j=1}\Big] || \leq 
\prod^{n}_{j=1} || \mathbb{T}_{j} || \big[|| f^{1}_0 ||^{2}_{H^{1}} \hspace{-0.2cm} + || \psi^{1}_{(x^{1}_1;y^{1}_1)} ||^{2}_{H^1}  +  || f^{2}_0 ||^{2}_{H^{2}} \hspace{-0.2cm} +  || \psi^{2}_{(x^{2}_1;y^2_1)} ||^{2}_{H^2}  \big]  + \nonumber \\
&\hspace{10cm}\sum^{n}_{j=2} \prod^{n}_{k=j} || \mathbb{T}_{k} || \big[  || \psi^{1}_{(x^{1}_j;y^{1}_j)} ||^{2}_{H^1} 
+ || \psi^{2}_{(x^{2}_j;y^{2}_j)} ||^{2}_{H^1} \big]  \label{Equation 10} 
\end{align}
\vspace{-0.9cm}
\end{figure*}
\section{Consistency of ColEst2L Algorithm}\label{Section 4}
To solve Problem \ref{Problem 6}, our goal is to uniformly bound the sequence $\{\bar{\mathbb{T}}_{n}\}$, so that a weakly convergent subsequence can be extracted. From Inequality \ref{Equation 10}, we note that the R.H.S depends on (i) the input data points from iteration $1$ to $n$; (ii) the product of the norm of the estimation operators from iteration $1$ to $n$. In the following subsection, (i) we define \textit{valid} input data sequences; (ii) find a subsequence of  $\{\bar{\mathbb{T}}_{n}\}$ whose product is bounded to ensure that $\{\bar{\mathbb{T}}_{n}\}$ is uniformly bounded.
\subsection{Uniform Boundedness of The Estimation Operators}
First, we define valid input data sequences. 
\begin{definition}\label{Definition 16}
   For a given initial condition $(f^1_0; f^2_0)$ of the estimation algorithm,  a sequence of input data points $\{(x^1_n; y^1_n)\}$ and  $\{(x^2_n; y^2_n)\}$ is said to be \textit{valid} if $c_{(f^{1}_0;f^{2}_0)}$ is finite, where $c_{(f^{1}_0;f^{2}_0)}$ is defined as
   \begin{align*}
       \hspace{-0.4cm} \sum^{\infty}_{n = 2} \dfrac{|| \psi^{1}_{(x^{1}_n;y^{1}_n)} ||^{2}_{H^1} + \psi^{2}_{(x^{2}_n;y^2_n)} ||^{2}_{H^2} }{|| f^{1}_0 ||^{2}_{H^{1}} \hspace{-0.25cm} + || \psi^{1}_{(x^{1}_1;y^{1}_1)} ||^{2}_{H^1}  +  || f^{2}_0 ||^{2}_{H^{2}} \hspace{-0.25cm} +  || \psi^{2}_{(x^{2}_1;y^2_1)} ||^{2}_{H^2}}.
   \end{align*}
   The set of valid sequences for a given initial condition $(f^1_0; f^2_0)$  denoted by $E_{(f^{1}_0; f^{2}_0)}$. 
\end{definition}
The above definition is to ensure that the second term in the R.H.S of Inequality \ref{Equation 10} is upper bounded for all $n$. 
\begin{lemma}\label{Lemma 17}
Let $\{\varrho^{1}_{n}\}$, $\{\varrho^{2}_{n}\}$ and $\{\varrho_{n}\}$ be sequences which diverge to $\infty$.  There exists a subsequence of $\{ \mathbb{T}_{n}\}$, $ \{\mathbb{T}_{n_{k_l}}\}$ such that $\underset{ l \to \infty} \lim \; \prod^{l}_{p=1} || \mathbb{T}_{n_{k_p}} ||  \leq   \underset{ l \in \mathbb{N}} \sup \; \prod^{l}_{p=1} || \mathbb{T}_{n_{k_p}} || \overset{\Delta}{=}  c_{M_1} < \infty$.  
\end{lemma}
\begin{proof}
Let $a_{n} = || \mathbb{T}_{n}||$. Then, $0 \leq a_{n} \leq || \hat{T} || \times || T(\varrho_n)  || \times ||\bar{T}(\varrho^{1}_{n}; \varrho^{2}_{n}) ||$. From Lemmas \ref{Lemma 10} and \ref{Lemma 12}, and, Definition \ref{Definition 13} it follows that $\{a_{n}\}$ is a bounded sequence of positive real numbers that has a subsequence which converges to a non-negative real number less than or equal to $1$. If every subsequence of $\{a_{n}\}$ converges to $0$, then $\{a_{n}\}$ converges to zero. We treat this as a pathological case, as it implies that irrespective of the input, the sequence of estimated functions converges to $\theta_{H^{1} \times H^2}$ which is not interesting. We consider the scenario where  there is subsequence $\{a_{n_k}\}$ that converges to a positive constant, $c_5$, less than or equal to 1. W.L.O.G we assume that the subsequence converges to 1 as the sequence $\{\frac{a_n}{c_5}\}$ converges to 1. There is a subsequence  of $\{a_{n_k}\}$, $\{a_{n_{k_l}}\}$ such that $ | a_{n_{k_l}} -1 | \leq \frac{1}{2^l}$. We define the sequence $b_{n} = \prod^{n}_{l=1}a_{n_{k_l}}$. We claim that the sequence $\{b_{n}\}$ converges. Indeed, consider $\ln(b_{n}) =\sum^{n}_{l=1}\ln(a_{n_{k_l}})$. Since  $\frac{-1}{2^l-1} \leq \ln(a_{n_{k_l}}) \leq \frac{1}{2^l}$, $|\ln(a_{n_{k_l}})| \leq \frac{1}{2^{l}-1}$. Since $\sum^{\infty}_{l=1}\frac{1}{2^{l}-1}$ converges, for every $\epsilon > 0$, $\exists N_{\epsilon}$ such that $\sum^{\infty}_{l=N_{e}} |\ln(a_{n_{k_l}})| < \epsilon$. Thus, $|\ln(b_{p}) - \ln(b_{q})| \leq \sum^{p}_{l=q} |\ln(a_{n_{k_l}})| < \epsilon \forall p, q \geq N_{\epsilon}$. Thus, the sequence $\{\ln(b_n)\}$ is Cauchy and converges to a real number. By continuity of the natural logarithm function, sequence $\{b_{n}\}$ converges to a real number lying in the interval $\big[\prod^{\infty}_{l=1} \frac{2^l-1}{2^l}, \prod^{\infty}_{l=1} \frac{2^l+1}{2^l}\big]$. From the construction of the subsequence $\{a_{n_{k_l}}\}$, it follows that,
\begin{align*}
    &\hspace{-0.7cm}\underset{ l \to \infty} \lim \; \prod^{l}_{p=1} || \mathbb{T}_{n_{k_p}} ||  \leq \underset{ l \in \mathbb{N}} \sup \; \prod^{l}_{p=1} || \mathbb{T}_{n_{k_p}} || \leq  \prod^{\infty}_{l=1} \frac{2^l+1}{2^l} < \infty. \\
    & \hspace{-0.7cm} c_{M_{2}} \overset{\Delta}{=} \underset{q \in \mathbb{N}} \sup \;\;  \underset{\substack{l \in \mathbb{N} \\ l \geq q}} \sup \;\;  \prod^{l}_{p=q} || \mathbb{T}_{n_{k_{p}}} || < \infty ,
\end{align*}
is also implied by above construction.
\end{proof}
We re-index the set of operators $\{\mathbb{T}_{n}\}^{n_{k_1}-1}_{n=1} \cup \{\mathbb{T}_{n_{k_l}}\}^{\infty}_{l=1}$ obtained from the above Lemma \ref{Lemma 17} (with $c_{M_1}$ and $c_{M_2}$ finite) to obtain a countable collection  of operators $\{\mathbb{T}_{n}\}^{\infty}_{n=1}$ for which $c_{M_1}$ and $c_{M_2}$ are finite. 
\begin{lemma}\label{Lemma 18}
Let $\{\varrho^{1}_{n}\}$, $\{\varrho^{2}_{n}\}$ and $\{\varrho_{n}\}$ be sequences which diverge to $\infty$. We define $\bar{\mathbb{T}}_{n}\Big |_{H^{1} \times H^{2}}: H^{1} \times H^{2} \to H^{1} \times H^{2}$ as, 
\begin{align*}
    \hspace{-0.9cm} \bar{\mathbb{T}}_{n}\Big |_{H^{1} \times H^{2}}(\cdot)[f^1;f^2] = \bar{\mathbb{T}}_{n}(\cdot)[f^1;f^2; \{\psi^{1}_{(x^{1}_j;y^{1}_j)}; \psi^{2}_{(x^{2}_j;y^{2}_j)}\}^n_{j=1}],
\end{align*}
where $\{\psi^{1}_{(x^{1}_j;y^{1}_j)}; \psi^{2}_{(x^{2}_j;y^{2}_j)}\}^n_{j=1}$ is a finite string of the first $n$ elements of \underline{any} sequence $\{\psi^{1}_{(x^{1}_j;y^{1}_j)}; \psi^{2}_{(x^{2}_j;y^{2}_j)}\}_{j \geq 1} \in E_{(f^{1}; f^{2})}$. Then, $\Big\{\bar{\mathbb{T}}_{n}\Big |_{H^{1} \times H^{2}}(\cdot)\big[\cdot\big]\Big\} $ is a uniformly equicontinuous sequence of linear operators. 
\end{lemma}
\begin{proof}
For a given initial estimate $(f^{1}_0; f^{2}_0)$ and an input data sequence from the set of valid input sequences corresponding to the initial condition, from Lemma \ref{Lemma 17} and Inequality \ref{Equation 10},  it follows that,
\begin{align}
    &\hspace{-0.7cm}|| \bar{\mathbb{T}}_n (\ldots)[\ldots] || \leq \Big[ c_{M_{1}} + c_{M_{2}} c_{(f^{1}_{0}, f^{2}_{0})}\Big] \times \nonumber \\
    &\hspace{-0.7cm}\Big[|| f^{1}_0 ||^{2}_{H^{1}} \hspace{-0.25cm} + || \psi^{1}_{(x^{1}_1;y^{1}_1)} ||^{2}_{H^1}  +  || f^{2}_0 ||^{2}_{H^{2}} \hspace{-0.25cm} +  || \psi^{2}_{(x^{2}_1;y^2_1)} ||^{2}_{H^2} \Big] \label{Equation 11} , 
\end{align}
$\forall n \in \mathbb{N}$, that is both the sequence of operators $\{\bar{\mathbb{T}}_{n}\}$ and $\{\bar{\mathbb{T}}_{n}\Big |_{H^{1} \times H^{2}}\}$ are pointwise bounded. However, from the uniform boundedness principle, \cite{royden2010real}, it follows that only the latter sequence (Remark \ref{Remark 19}), is uniformly bounded, $\exists c_{\mathbb{T}} >0 $ such that $|| \bar{\mathbb{T}}_{n}\Big |_{H^{1} \times H^{2}} || \leq c_{\mathbb{T}} , \; \forall n \in \mathbb{N}$. Given $\epsilon >0$, let $\delta < \frac{\epsilon}{ c_{\mathbb{T}}}$. Then, $|| f - g|| < \delta $ implies that
\begin{align*}
    &\hspace{-0.8cm}|| \bar{\mathbb{T}}_{n}\Big |_{H^{1} \times H^{2}}(f) - \bar{\mathbb{T}}_{n}\Big |_{H^{1} \times H^{2}}(g) || \leq \\
    &|| \bar{\mathbb{T}}_{n}\Big |_{H^{1} \times H^{2}} || \times || f- g || <  c_{\mathbb{T}} \times \frac{\epsilon}{ c_{\mathbb{T}}} < \epsilon, \;  \forall n \in \mathbb{N}.
\end{align*}
\end{proof}
\begin{remark}\label{Remark 19}
For invoking the uniform boundedness principle as proved in \cite{royden2010real}, a common domain for the sequence of operators, $\{\bar{\mathbb{T}}_{n}\}$ is needed. One candidate would be $H^{1} \times H^2 \times (H^{1} \times H^{2})^\infty$. However,  the bound presented in Inequality \ref{Equation 11} does not hold for every input sequence in $(H^{1} \times H^{2})^\infty$. An alternative candidate would be, 
\begin{align*}
&\hspace{-0.7cm} \mathcal{X}_{\mathbb{T}} = \{(f^1;f^2; \{\psi^{1}_{(x^{1}_j;y^{1}_j)}; \psi^{2}_{(x^{2}_j;y^{2}_j)}\}_{j \geq 1}  \in H^{1} \times H^2 \\ 
&\hspace{0.2 cm}\times (H^{1} \times H^{2})^\infty :   \{\psi^{1}_{(x^{1}_j;y^{1}_j)}; \psi^{2}_{(x^{2}_j;y^{2}_j)}\}_{j \geq 1} \in E_{f^{1}, f^{2}} \}.
\end{align*}
Note that $ \mathcal{X}_{\mathbb{T}}$ is not necessarily a subspace. If it is a complete metric space with a suitable norm, then from Theorem 6 chapter 10 of \cite{royden2010real}, it follows that $\exists \; O \subset \mathcal{X}_{\mathbb{T}}$, open and $c_{M_3}$ such that $ || \bar{\mathbb{T}}_{n}(\cdot)[f] || \leq c_{M_3}, \forall f \in O, \; \forall n \in \mathbb{N}$. This however cannot be extended to entire $ \mathcal{X}_{\mathbb{T}}$ as it is not a vector space.  Hence, in Lemma \ref{Lemma 18}, we define a sequence of operators with domain restricted to $H^{1} \times H^{2}$. The output of the operators will change as the input data sequence  changes. However, irrespective of the inputs,  the outputs are uniformly bounded with respect to the initial estimates. 
\end{remark}
Thus, through Definition \ref{Definition 16} and Lemma \ref{Lemma 18} we have solved items 1 and 2 in Problem \ref{Problem 6}.  
\subsection{Main Result}
In this section, we define the estimation operator for the algorithm for arbitrarily large $n$ and solve Problem \ref{Problem 6}. 
\begin{definition}\label{Definition 20}[\underline{\textit{Estimation Operator for Algorithm}}]
Given $(f^{1}_0;f^{2}_0)$, $\{\varrho^{1}_n; \varrho^{2}_n; \varrho_n\}$, and a sequence of valid data points $\{(x^{1}_n; y^{1}_{n})\}, \{(x^{2}_n;y^{2}_{n})\} \in E_{(f^{1}_0;f^{2}_0)}$, the estimation operator for the algorithm, $\mathbb{T}:  \mathcal{X}_{\mathbb{T}} \to H^{1} \times H^{2}$ is recursively  defined as follows
\begin{align*}
    &\hspace{-0.7cm} \mathbb{T}(\{(\varrho^{1}_n;\varrho^{2}_n;\varrho_n)\})\Big[f^{1}_0;f^{2}_0; \{\psi^{1}_{(x^{1}_n;y^{1}_n)}; \psi^{2}_{(x^{2}_n;y^{2}_n)} \}_{n \geq 1}\Big]  \\ 
    & \hspace{-0.9cm} = \hspace{-0.1cm}  \underset{ n \to \infty } \lim \bar{\mathbb{T}}_{n}(\{(\varrho^{1}_k;\varrho^{2}_k; \varrho_k)\}^{n}_{k=1})\Big[f^{1}_{0},f^{2}_{0};\{\psi^{1}_{(x^{1}_j;y^{1}_j)}; \psi^{2}_{(x^{2}_j;y^{2}_j)}\}^{n}_{j=1}\Big],
\end{align*}
where $\bar{\mathbb{T}}_{n}$ is defined in Definition \ref{Definition 15}.  
\end{definition}
In the proof of the following result, we suppress the parameter related arguments in the notation of operators $\bar{\mathbb{T}}_{n}(\cdot)[\cdot]$ and $\mathbb{T}(\cdot)[\cdot]$ and instead use the notation $\bar{\mathbb{T}}_{n}(\cdot)$ and $\mathbb{T}(\cdot)$,  where the arguments are the functions that they operate on. Further, only the arguments which correspond to the initial condition is mentioned, the arguments corresponding to the terms in the input data sequence are suppressed. At a given initial condition, it is understood that the input data sequence is \textit{any} valid sequence corresponding to that initial estimate. 
\begin{theorem}\label{Theorem 21}
The estimation operator is well defined in the following sense; there exists a subsequence of $\Big\{\bar{\mathbb{T}}_{n}\Big(f^{1}_0;f^{2}_0;\\ \{\psi^{1}_{(x^{1}_j;y^{1}_j)}; \psi^{2}_{(x^{2}_j;y^{2}_j)}\}^n_{j=1}\Big)\Big\}$ which strongly converges to $(f^{1,*}; \\f^{2,*})\in H^{1} \times H^{2}$ for any initial condition $(f^{1}_0;f^{2}_0) \in H^{1} \times H^{2}$ and input data sequence $\{\psi^{1}_{(x^{1}_j;y^{1}_j)}; \psi^{2}_{(x^{2}_j;y^{2}_j)}\}_{j \geq 1}) \in E_{(f^{1}_0;f^{2}_0)}$. $\mathbb{T}$ is linear with respect to the initial estimates and is a bounded operator.
\end{theorem}
\begin{proof}
For any $g \in H^{1} \times H^{2}$, to prove that the sequence (or a subsequence of it ), $\{\bar{\mathbb{T}}_{n}(g)\}$ converges weakly to $\mathbb{T}(g)$, we have to prove $\{\Psi(\bar{\mathbb{T}}_{n}(g))\}$ converges to $\Psi(\mathbb{T}(g))$ for all $\Psi \in  (H^{1} \times H^{2})^*$. Since $ (H^{1} \times H^{2})^*$ is isomorphic to  $H^{1} \times H^{2}$, for every $\Psi \in (H^{1} \times H^{2})^*, \; \exists f_{\Psi} \in H^{1} \times H^{2}$ such that $\Psi(g) = \langle g, f_{\Psi} \rangle_{H^{1} \times H^{2}}, \; \forall g \in H^{1}\times H^{2}$. Following this reasoning, we define,
\begin{align*}
    \Psi_{n}(f)[g] \overset{\Delta}{=} \langle \bar{\mathbb{T}}_{n}(g), f \rangle_{H^1 \times H^2}, f \in H^{1} \times H^{2}. 
\end{align*}
Pointwise convergence on $ (H^{1} \times H^{2})^*$ is modified to pointwise convergence on $H^{1} \times H^2$. Since $H^{1} \times H^{2}$ is finite dimensional space, it is separable. Let $\{\psi_n\}$ be an enumeration of a countable dense set, $\mathbb{D}$, in $H^{1} \times H^{2}$. Consider the sequence $\Psi_{n}(f)[\psi_1]$. By Lemma \ref{Lemma 18} and the CBS inequality, it follows that,
\begin{align*}
    | \langle \bar{\mathbb{T}}_{n}(\psi_1), f \rangle_{H^1 \times H^2} | \leq c_{\mathbb{T}} \;  || \psi_1 ||_{H^1 \times H^2} \; || f ||_{H^1 \times H^2}, 
\end{align*}
i.e., the sequence $\{\Psi_{n}(\cdot)[\psi_1]\}$ is bounded sequence of bounded linear functionals on a separable Hilbert space. By \textit{Helley's theorem} (refer chapter 8, \cite{royden2010real}) there exists a subsequence, a strictly increasing sequence of integers, $\{s(1,n)\}$, such that $\{\Psi_{s(1,n)}(\cdot)[\psi_1]\}$ converges pointwise to $\Psi^{*}_{\psi_1}(\cdot) \in (H^{1} \times H^{2})^*$. By the \textit{Riesz- Fr\'echet Representation Theorem} (refer  chapter 16, \cite{royden2010real}), $\exists  \psi^{*}_1 \in H^{1} \times H^{2}$ such that,
\begin{align*}
\Psi^{*}_{\psi_1}(f) =  \Psi^{*}(f)[ \psi^{*}_1] \overset{\Delta}{=} \langle  \psi^{*}_1, f \rangle, \forall f \in H^{1} \times H^{2} 
\end{align*}
Since the convergence, $\Psi_{s(1,n)}(\cdot)[\psi_1] \to \Psi^{*}(\cdot)[ \psi^{*}_1] $  is pointwise on the dual space , $(H^{1} \times H^{2})^*$,  the $ \psi^{*}_1$ represents an element in dual of $(H^{1} \times H^{2})^*$, which is $(H^{1} \times H^{2})$. Hence, it appears in ``linear" term of the inner product rather than the ``antilinear term." Thus, $ \bar{\mathbb{T}}_{s(1,n)}(\psi_1) \rightharpoonup \psi^{*}_1$. Since all the inner products and norms in this proof are evaluated in $H^1 \times H^2$, going forward this is suppressed.  

By the same argument, since the sequence $\langle \bar{\mathbb{T}}_{n}(\psi_2), f \rangle$ is bounded, there exists a subsequence of $\{s(1,n)\}$, $\{s(2,n)\}$, such that $\bar{\mathbb{T}}_{s(2,n)}(\psi_2) \rightharpoonup \psi^{*}_2, \psi^{*}_2 \in H^{1} \times H^{2}$. We can inductively continue this process to obtain a strictly increasing sequence of integers $\{s(j,n)\}$ which is a subsequence of $\{s(j-1,n)\}$ such that $\bar{\mathbb{T}}_{s(j,n)}(\psi_k) \rightharpoonup \psi^{*}_k$. For each $j$, we define $\mathbb{T}(\psi_j)$ as $\psi^{*}_j$. By \textit{Cantor's diagonalization argument} (refer chapter 8, \cite{royden2010real}), consider the subsequence of operators, $\bar{\mathbb{T}}_{n_k}(\cdot)$ , where $n_{k} =s(k,k)$ for every index $k$. For each $j$, the subsequence, $\{n_{k}\}^{\infty}_{k=j}$ is a subsequence of the $j$th subsequence of integers chosen before and thus $ \bar{\mathbb{T}}_{n_k}(\psi_{j}) \overset{k}{\rightharpoonup} \mathbb{T}(\psi_j)\; \forall j \in \mathbb{N}.$ Thus, $\{\bar{\mathbb{T}}_{n_k}(\psi) \}$ converges weakly to  $\mathbb{T}(\psi)$ on $\mathbb{D}$.

Let $g$ be any function in $H^{1} \times H^{2}$. We claim that the sequence $\{\Psi_{n_k}(f)[g]\}$ is Cauchy for every $f$, where $ \Psi_{n_k}(f)[g] = \langle \bar{\mathbb{T}}_{n_k}(g), f \rangle$. Indeed, since the sequence operators $\{\bar{\mathbb{T}}_{n_k} \Big|_{H^{1} \times H^{2}} \}$ is equicontinuous, $\forall\epsilon >0, \exists  \delta > 0$ such that, $|| \bar{\mathbb{T}}_{n_k}(h) - \bar{\mathbb{T}}_{n_k}(g)|| < \frac{\epsilon}{3 \times ||f||}$ for all $h$ such that $|| h -g || <\delta$, all indices $n_k$, and, any valid input data signals at $h,g$. This implies that $ | \langle \bar{\mathbb{T}}_{n_k}(h), f \rangle  - \langle \bar{\mathbb{T}}_{n_k}(g), f \rangle | \leq || \bar{\mathbb{T}}_{n_k}(h) - \bar{\mathbb{T}}_{n_k}(g)||\; ||f|| <  \frac{\epsilon}{3}$ $\forall h$ such that $|| h -g || <\delta$ and all indices $n_k$. Since $\mathbb{D}$ is dense in $H^{1} \times H^{2}$, there exists $\psi \in \mathbb{D}$ such that $||\psi - g|| < \delta$. Since the sequence, $\{ \langle \bar{\mathbb{T}}_{n}(\psi), f \rangle\}$ is cauchy, there exists $N_{\epsilon}$ such that $|\langle \bar{\mathbb{T}}_{n_p}(\psi), f \rangle  - \langle \bar{\mathbb{T}}_{n_q}(\psi), f \rangle  | < \frac{\epsilon}{3}, \; \forall p, q \geq N_{\epsilon}$. Thus, for all $p,q \geq N_{\epsilon}$,
\begin{align*}
&\hspace{-0.5cm}|\Psi_{n_p}(f)[g] - \Psi_{n_q}(f)[g]| \leq \\
&\hspace{-0.5cm}|\Psi_{n_p}(f)[g] - \Psi_{n_p}(f)[\psi]| +
|\Psi_{n_p}(f)[\psi] - \Psi_{n_q}(f)[\psi]|+ \\
&\hspace{-0.5cm}|\Psi_{n_q}(f)[\psi] - \Psi_{n_q}(f)[g]| < \frac{\epsilon}{3} + \frac{\epsilon}{3} + \frac{\epsilon}{3} = \epsilon.
\end{align*}
\begin{center}
\begin{tikzpicture}
  \matrix (m) [matrix of math nodes, row sep=5em,
    column sep=5em]{
     \Psi_{n_p}(f)[g] & \Psi_{n_p}(f)[\psi] \\
     \Psi_{n_q}(f)[g] & \Psi_{n_q}(f)[\psi] \\};
     \path[-stealth]
     (m-1-1)	edge node [fill=white] {$\frac{\epsilon}{3}$} 	    (m-1-2)
		edge node [fill=white] {$\epsilon$} 		 (m-2-1)
  
    (m-1-2) edge node [fill=white] {$\frac{\epsilon}{3}$}		 (m-2-2)
             
    (m-2-2) edge node [fill=white] {$\frac{\epsilon}{3}$}		(m-2-1)  ;    
\end{tikzpicture}
\end{center}
Thus, $\{\Psi_{n_k}(f)[g]\}$ converges to real number which denote by $\Psi^*(f)[g]$. From the linearity of each functional in the sequence $\{\Psi_{n_k}(\cdot)[g]\}$, it follows that $\Psi^*(\cdot)[g]$ is linear. Since, $| \Psi_{n_k}(f)[g] | \leq c_{\mathbb{T}} ||g||\; ||f||, \forall n_{k}$, it follows that $\underset{k \to \infty} \lim  | \Psi_{n_k}(f)[g] | = |\underset{k \to \infty} \lim   \Psi_{n_k}(f)[g]| = \hspace{-3pt} |\Psi^*(f)[g]|  \leq c_{\mathbb{T}}||g||\; ||f||$. Thus, $\Psi^*(\cdot)[g] \in (H^{1} \times H^{2})^*$. By the \textit{Riesz- Fr\'echet Representation Theorem}, $\exists  \psi^{*}_g \in H^{1} \times H^{2}$ such that $\Psi^{*}(f)[g] = \langle  \psi^{*}_g, f \rangle, \forall f \in H^{1} \times H^{2}$. Thus, $\langle \bar{\mathbb{T}}_{n_k}(g), f \rangle \to \langle \psi^{*}_g, f \rangle, \forall f \in H^{1} \times H^{2}$.  We let $\mathbb{T}(g) = \psi^{*}_g$ and from the arguments presented, we conclude that $\bar{\mathbb{T}}_{n_k}(g) \rightharpoonup \mathbb{T}(g), \; \forall g \in H^{1} \times H^{2}$. 

Let $\{\hat{\varphi}_{j}\}^{|\mathcal{I}^1| + |\mathcal{I}^2|}_{j=1} \hspace{-3pt}$ be an orthonormal basis for $H^{1} \times H^{2}$ obtained by applying the \textit{Gram–Schmidt orthonormalization process} to $\{\varphi^{1}_{j} ; \theta^2 \}_{j \in \mathcal{I}^{1}} \cup \{\theta^1 ; \varphi^{2}_{j} \}_{j \in \mathcal{I}^{2}}$. By the weak convergence result, $\Psi_{n_k}(\hat{\varphi}_{j})[g] \to \Psi^*( \hat{\varphi}_{j})[g] , \; \forall g \in H^{1} \times H^{2}$ and $j$, where $\Psi_{n_k}(f)[g] \overset{\Delta}{=} \langle \bar{\mathbb{T}}_{n_k}(g),f\rangle $ and $ \Psi^*(f)[ g ] \overset{\Delta}{=} \langle \mathbb{T}(g), f\rangle$ as defined before. This implies that, for given $g \in H^{1} \times H^2$,  $\forall \epsilon >0$, $\exists N_{\epsilon,g}$ such that $|\Psi_{n_k}(\hat{\varphi}_{j})[g] - \Psi^*(\hat{\varphi}_{j})[ g] | < \frac{\epsilon}{|\mathcal{I}^1| + |\mathcal{I}^2|}, \forall j, \forall k \geq N_{\epsilon,g}$. Thus, $\forall \epsilon >0$, $\exists N_{\epsilon,g}$ such that 
\begin{align*}
&\hspace{-0.9cm}|| \bar{\mathbb{T}}_{n_k}(g) - \mathbb{T}(g) || = \Big|\Big| \hspace{-0.1cm} \sum^{|\mathcal{I}^1| + |\mathcal{I}^2|}_{j=1} \hspace{-0.1cm} \Big [ \Psi_{n_k}(\hat{\varphi}_{j})[g] - \Psi^*(\hat{\varphi}_{j})[g] \Big] \hat{\varphi}_j \Big|\Big| \\
&\hspace{-0.9cm} \overset{(a)}{\leq}  \sum^{|\mathcal{I}^1| + |\mathcal{I}^2|}_{j=1}   | \Psi_{n_k}(\hat{\varphi}_{j}) [g] - \Psi^*(\hat{\varphi}_{j})[g] | < \epsilon.
\end{align*}
Inequality $(a)$ follows from triangle inequality and the normalization of $\{\hat{\varphi}_{j} \}$ i.e., $||\hat{\varphi}_{j} || =1$. Thus, $\{ \bar{\mathbb{T}}_{n_k}(g) \}$ converges to $\mathbb{T}(g)$ in strong topology for all $g \in H^{1} \times H^{2}$. From the linearity of the sequence of operators, $\{ \bar{\mathbb{T}}_{n_k} \}$, in the initial estimates, it follows that  $\mathbb{T}$ is linear in the initial estimates. $|| \mathbb{T}(f) || = \underset{k \to \infty}  \lim ||\bar{\mathbb{T}}_{n_k}(f) || \leq \underset{k \to \infty} \lim  ||\bar{\mathbb{T}}_{n_k} || \; || f || \leq c_{\mathbb{T}} ||f ||$. Thus,  $\mathbb{T}$ is bounded in the initial estimate given any corresponding valid input data sequence. 
\end{proof}
\begin{remark}\label{Remark 22}
We note that weak convergence result in the proof of theorem holds for separable Hilbert spaces, not just finite dimensional or compact Hilbert spaces. The proof of extending weak convergence to strong convergence holds only for finite dimensional spaces.
\end{remark}
The above proof can be analyzed as follows. If we were to prove that every sequence of estimated functions is bounded  with a different bound, by the \textit{Bolzano–Weierstrass Theorem}, there exists a subsequence of estimated functions which converges in norm. Further, for different sequences of estimated functions, there could different subsequences which are converging. In the case of an infinite dimensional Hilbert space, the boundedness of the sequence of estimated functions would only imply a weakly convergent subsequence. 

However the operator theoretic framework presents a stronger result: Theorem \ref{Theorem 21} is a ``stronger" statement in the sense that for every sequence of learned functions, there is a common subsequence that converges across every sequence to function in the product of the knowledge spaces. For an initial condition $f$ and valid input data sequence , consider the sequence of functions, $\phi_{n} : (H^{1} \times H^{2})^* \to \mathbb{R} $ as $\phi_{n}(\Psi) = \langle \bar{\mathbb{T}}_{n}(f), T (\Psi) \rangle_{H^{1} \times H^{2}}$ where $T :(H^{1} \times H^{2})^* \to H^{1} \times H^{2} $ is an isometric anti-linear isomorphism that assigns to any $\Psi \in (H^{1} \times H^{2})^*$,  a representing element in $H^{1} \times H^{2}$. $\{\phi_n\}$ is a sequence of equicontinuous functions on the dual space which are pointwise bounded which implies that there exists  a subsequence which converges pointwise to  real valued function on the dual space. 

The topology of the space of initial conditions, specifically its separability, enables us to find a subsequence which converges pointwise for every sequence $\{\phi_{n}\}$ generated by choosing an initial condition $f$  from a countable dense set. Equivalently,  we find a subsequence which weakly converges for every sequence of learned functions initiated from a countably dense set of space of initial conditions. Then, we prove that this subsequence works for every sequence $\{\phi_{n}\}$  generated by varying the initial condition over its entire space. This is achieved by utilizing the separability of the space and the equicontinuity of $\{\phi_{n}\}$. Linearity and equicontinuity of $\{\phi_n\}$ enables us to show that the function to which the common subsequence converges to is also linear and continuous for every initial condition; thus representing a element in the dual of the dual space, i.e, the original Hilbert space. This equivalently translates to a common subsequence converging weakly for every sequence of estimated functions.
\begin{corollary}\label{Corollary 23}
The estimation algorithm presented in algorithm 1 is strongly consistent. 
\end{corollary}
\begin{proof}
Follows from Definition \ref{Definition 4} and Theorem \ref{Theorem 21}. 
\end{proof}
An alternative approach to show consistency, would have been to require that the norm of the estimation operators at every iteration is less than $1$, i.e, a contraction map across some suitable space. This would lead to proving a stronger result in a sense  that the entire sequence of learned functions converges. By allowing the norms of the operators, $\{\bar{\mathbb{T}}_n\}$, to be in the neighbourhood of $1$ instead of strictly less than  $1$, we lose the convergence of the entire sequence. However, we have shown that there is a common subsequence that converges for every sequence of learned functions, which would in a sense correspond to regaining the stronger result with some suitable re-indexing. In following, we prove that given an input data sequence, there exists an initial condition which is invariant to the estimation process. 
\subsection{Invariance to Estimation Algorithm}
The following definition is a counterpart to Definition \ref{Definition 16}.  For given input sequence, we define the set of valid initial conditions. For any valid initial condition and the given input, we consider the restriction of the estimation operator in Definition \ref{Definition 20} and subsequently study its continuity. 
\begin{definition}\label{Definition 24}
    Given an input data sequence, $s = \{(x^{1}_{n}; y^{1}_{n}) , \\ (x^{2}_{n}; y^{2}_{n})\}$, let $E_{s} \subset H^{1} \times H^{2}$, $E_{s} \neq \varnothing $,  be the set of initial conditions for which $s$ is a valid input. Then, we define $\mathbb{T}_s : E_s  \to H^{1} \times H^{2}$ as follows. For $(f^{1}; f^{2}) \in E_{s}$, 
    \begin{align*}
        &\hspace{-0.6cm}\mathbb{T}_{s}(f^{1};f^2) = \\
        &\hspace{0.3cm}\mathbb{T}(\{(\varrho^{1}_n;\varrho^{2}_n;\varrho_n)\})\Big[f^{1};f^{2}; \{\psi^{1}_{(x^{1}_n;y^{1}_n)}; \psi^{2}_{(x^{2}_n;y^{2}_n)} \}_{n \geq 1}\Big]. 
    \end{align*}
\end{definition}
\begin{lemma}\label{Lemma 25}
   $\mathbb{T}_s(\cdot)$ is continuous. 
\end{lemma}
\begin{proof}
First, we note that even though $\mathbb{T}$ is linear, $\mathbb{T}_s(\cdot)$ is not linear as its domain is a subset of a vector space, not a vector space by itself.  For $f \in E_{s}$, from Theorem \ref{Theorem 21}, it follows that there is a subsequence of $\{\bar{\mathbb{T}}_{n}(f, \ldots)\}$, $\{\bar{\mathbb{T}}_{n_k}(f, \ldots)\}$, which converges strongly for any $f \in E_s$. With a slight abuse of notation, we will refer to $\bar{\mathbb{T}}_{n_k}(f, \ldots)$ as $\bar{\mathbb{T}}_{n_k}(f)$ as the second argument is fixed to input data sequence $s$. For $f \in E_{s}$, from Definition \ref{Definition 16} it follows that, $\exists \delta_f >0$ such that $B_{\delta_f}(f) \subset E_{s}$ where $B_{\delta}(f)  \overset{\Delta}{=}\{g \in H^{1} \times H^{2} : || g- f|| < \delta\}$. Given the valid input data sequence, the linearity and uniform boundedness (Lemma \ref{Lemma 18}) of $\{\bar{\mathbb{T}}_{n_k}\}$ implies that given $\epsilon > 0, \exists\;\delta,  0 < \delta \leq  \delta_f$, such that, for any $g, h \in B_{\delta_f}(f)$
\begin{align*}
    || g - h || < \delta \implies || \bar{\mathbb{T}}_{n_k}(g) - \bar{\mathbb{T}}_{n_k}(h) || < \frac{\epsilon}{3}, \forall k
\end{align*}
Since $B_{\delta_f}(f)$ is compact, given $\delta >0$ as above, there exists $\{\varphi_{f,j, \delta}\}^{n_f}_{j=1}$ such that $B_{\delta}(f) \subset \bigcup^{n_f}_{j=1}B_{\delta}(\varphi_{f,j, \delta})$.  Since $\{\bar{\mathbb{T}}_{n_k}(\varphi_{f,j, \delta})\}$ is strongly Cauchy for every $\varphi_{f,j, \delta}$,  given $\epsilon > 0, \exists N_{\epsilon} $ such that $|| \bar{\mathbb{T}}_{n_p}(\varphi_{f,j, \delta}) - \bar{\mathbb{T}}_{n_q}(\varphi_{f,j, \delta}) || < \frac{\epsilon}{3}, \forall p,q \geq N_{\epsilon}, j= 1, \ldots, n_f$. Thus,  given $\epsilon >0$, $\exists \; p, q \geq N_{\epsilon}$, such that $\forall \; g \in B_{\delta_f}(f)$, 
\begin{align*}
&\hspace{-0.7cm}||  \bar{\mathbb{T}}_{n_p}(g) - \bar{\mathbb{T}}_{n_q}(g) || \leq  ||  \bar{\mathbb{T}}_{n_p}(g) - \bar{\mathbb{T}}_{n_p}(\varphi_{f,j, \delta}) || + \\
&\hspace{-0.7cm}||  \bar{\mathbb{T}}_{n_p}(\varphi_{f,j, \delta}) - \bar{\mathbb{T}}_{n_q}(\varphi_{f,j, \delta}) ||  + ||  \bar{\mathbb{T}}_{n_q}(\varphi_{f,j, \delta}) - \bar{\mathbb{T}}_{n_q}(g) || < \epsilon.
\end{align*} 
 Thus, the sequence $\{\bar{\mathbb{T}}_{n_p}(g)\}$ is uniformly Cauchy for every $g \in B_{\delta_f}(f)$ and uniformly converges to $\mathbb{T}_s(g)$. Let $\{f_n\}$ be a sequence in $B_{\delta_f}(f)$ which converges to $f$. Then, 
 \begin{align*}
     \underset{ j \to \infty}{\lim} \mathbb{T}_s (f_j) =  \underset{ j \to \infty}{\lim} \; \underset{k \to \infty}{\lim}\bar{\mathbb{T}}_{n_k} (f_j) \overset{(a)}{=} \underset{k \to \infty}{\lim}\;  \underset{ j \to \infty}{\lim} \bar{\mathbb{T}}_{n_k} (f_j) \\
     =  \underset{k \to \infty}{\lim}\; \bar{\mathbb{T}}_{n_k} (f) = \mathbb{T}_{s}(f).
     \vspace{-0.5cm}
 \end{align*}
 where equality $\overset{(a)}{=}$ follows from the \textit{Moore-Osgood Theorem} with $\underset{k \to \infty}{\lim} || \bar{\mathbb{T}}_{n_k}(f_{j}) - \mathbb{T}_s(f_j) ||= 0$ uniformly and $\underset{j \to \infty}{\lim}|| \bar{\mathbb{T}}_{n_k}(f_{j}) - \bar{\mathbb{T}}_{n_k}(f) ||= 0$.
\end{proof}
 An alternative approach would have been to show that the sequence $\{ \bar{\mathbb{T}}_{n_k} \}$ of operators is Cauchy in space of linear continuous operators from $H^{1} \times H^{2}$ to  $H^{1} \times H^{2}$, $\mathcal{L}(H^{1} \times H^{2}, H^{1} \times H^{2})$, equipped with the operator norm. Then, the completeness of this space with the associated norm, would imply that the sequence converges to a continuous operator on $H^{1} \times H^{2}$. However, this approach could not be used as the domain of $\mathbb{T}_{s}$ is $E_s$ and not $H^{1} \times H^2$. 
\begin{corollary}\label{Corollary 26}
Suppose $E_s$ is convex and compact, and,  $\bar{\mathbb{T}}_s = \frac{\mathbb{T}_s}{|| \mathbb{T}_s||}$ is such that $\bar{\mathbb{T}}_s(E_s) \subset E_s$. Then, $\exists\; f \in E_s$ such that $\bar{\mathbb{T}}_s(f)  = f$.
\end{corollary}
\begin{proof}
The result follows from Lemma \ref{Lemma 25} and \textit{Brouwer's fixed-point theorem}. 
\end{proof}
Using Definition \ref{Definition 16}, we can derive conditions under which the set $E_s$ is convex and compact. However, the condition $\bar{\mathbb{T}}_s(E_s) \subset E_s$ requires further investigation; the unnormalized estimation operator leading to only scaling of the domain needs to be studied. 
\section{Numerical Example}\label{Section 5}
\begin{figure}
\begin{center}
\includegraphics[scale=0.54]{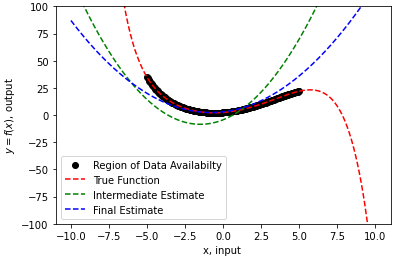}
\includegraphics[scale=0.54]{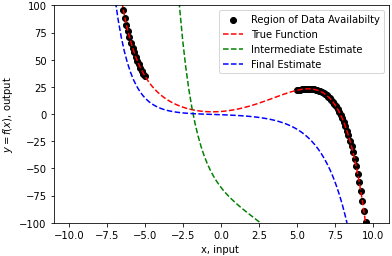}
\caption{True function and estimated functions at: (above) Agent 1 and (below) Agent 2.}
\label{Figure 4}
\vspace{-0.9cm}
\end{center}
\end{figure}
In this section, we consider an example to demonstrate the ColEst2L algorithm mentioned in subsection \ref{Algorithm Description}. True data is generated by considering a real valued function whose inputs are real values, obtained through a linear combination of polynomials and exponentials. Each agent receives noisy version of true data, i.e., noise added to the output data.

Agent 1 considers the features, $\varphi^1_{1}(x) =1, \varphi^1_{2}(x) =x$ and $\varphi^1_3(x) = x^{2}$, while Agent 2 considers the features $\varphi^2_{1}(x) =\exp(-x), \varphi^2_{2}(x) =\exp(x)$. Thus, the kernel corresponding to Agent $1$ is $K^{1}(x,y) = 1 + xy + x^2 y^2$ and to Agent $2$ is $K^{2}(x,y) = \exp(-x -y) + \exp(x+y)$. The domain of the input data for Agent 1 is considered to be $[-5,5]$, while for Agent $2$ it is considered to be $[-10,-5] \cup [5,10]$. At time step $n$, after collecting data point $(x^{i}_n, y^{i}_n)$, each agent solves $(P1)^{i}_n$ to obtain $f^{i}_{n}$ which is then uploaded to the fusion space. 
\begin{figure}
\begin{center}
\includegraphics[scale=0.54]{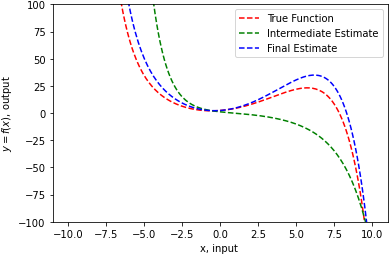}
\caption{True function and estimated functions at Fusion Center.}\label{Figure 5}
\vspace{-1.0cm}
\end{center}
\end{figure}

The fusion space corresponds to the RKHS generated by the kernel, $K(x,y) = 1 + xy + x^2 y^2 +  \exp(-x -y) + \exp(x+y)$. We note that the feature vectors for a given agent and across agents are linearly independent, which implies that $H$ has dimension $5$. We chose $\bar{x}^{1}_{1} = 0, \bar{x}^{1}_{2}= 2, \bar{x}^{1}_{3} = 4, \bar{x}^{1}_{4} = -2, \bar{x}^{1}_{5} =-4$ and $\bar{x}^{2}_{1} = 1, \bar{x}^{2}_{2}= 3, \bar{x}^{2}_{3} = 5, \bar{x}^{2}_{4} = -1, \bar{x}^{2}_{5} =-3$ to solve the fusion problem, subsection \ref{Fusion at Fusion Center}. After estimating $\mathbf{\hat{Y}}^{i}_{n}$ as in Proposition \ref{Proposition 3}, the fusion problem in $(P3)_{n}$ is solved. To find the download operator, first, we choose the set of basis vectors for the space $H$ as $\varphi_1(x)=1, \varphi_2(x) = x, \varphi_3(x) = \sqrt{2}x^2 ,\varphi_4(x) =\exp(-x), \varphi_{5}(x) =\exp(x)$. With these basis vectors, the coefficients for $K^{1}(\cdot,y)$ are $[1;y;y^2; 0; 0]$, and for $K^{2}(\cdot,y)$ are $[0,0,0,\exp(-x),\exp(x)]$. Thus, the matrix representation for $\bar{L}^{i}$ is obtained as follows:
\begin{align*}
&\bar{L}^{1}(\varphi_1)(y)  \hspace{-2pt} = \hspace{-2pt} \langle \varphi_1(\cdot), K^{1}(\cdot,y) \rangle_{H} \hspace{-2pt} =\hspace{-3pt} \langle [1,0,0,0,0],\\
&[1,y, y^2, 0, 0] \rangle_{\mathbb{R}^5} =1. \;\hspace{-2pt}\bar{L}^{1}(\varphi_2)(y) \hspace{-2pt}  = \hspace{-2pt} \langle [0,1,0,0,0],\\ 
&[1,y, y^2, 0,0] \rangle_{\mathbb{R}^5} \hspace{-2pt} = \hspace{-2pt}y. \; \hspace{-2pt}\bar{L}^{1}(\varphi_3)(y) = \langle [0,0,1,0,0], \\
&[1,y, y^{2}, 0,0] \rangle_{\mathbb{R}^5} =y^{2}.\; \bar{L}^{1}(\varphi_4)(y) = 0.\\
&  \bar{L}^{1}(\varphi_5)(y) = 0, L^{1}_{M}= \sqrt{L^{1}_{M}}= 
\begin{bmatrix}
1 & 0 & 0 & 0 &0\\
0 & 1 & 0 & 0 &0\\
0 & 0 & 1 & 0 &0\\
0 & 0 & 0 & 0 &0\\
0 & 0 & 0 & 0 &0\\
\end{bmatrix}, \; \\
&L^{2}_{M}= \sqrt{L^{2}_{M}}= 
\begin{bmatrix}
0 & 0 & 0 & 0 &0\\
0 & 0 & 0 & 0 &0\\
0 & 0 & 0 & 0 &0\\
0 & 0 & 0 & 1 &0\\
0 & 0 & 0 & 0 &1\\
\end{bmatrix}.
\end{align*}
With this setup, simulations were run and the results are demonstrated in Figure \ref{Figure 4} and Figure \ref{Figure 5}. The algorithm ran for approximately $1000$ iterations.  In each of the figures, the true function, function estimated at an intermediate iteration,  and the final estimate are plotted. In the figures pertaining to the agents, the segment of the true function which is accessible to the agents for data collection is also marked.  Our observations from the figures are as follows. Since the kernel for Agent 1 quadratic, the curve estimated by it is quadratic. The final estimate of Agent 1 partially overlaps with the true curve. Since the feature maps for Agent 2 are exponentials, the same is reflected in its estimates. However, neither of them are able to capture the true function ``completely". The final estimate in the fusion space captures the true function with minimum norm of the error, i.e., with minimum $|| f_{n^*} - f^{*}||$ among all iterations, where $n^*$ is the final iteration and $f^{*}$ is the true iteration. In reality, $f^*$ is not unknown and it would not be possible to verify the same. 

A limitation of the algorithm which was observed is numerical instability with respect to the parameters. It was observed that for a fixed data sequence, initial estimates, and, iteration $n$, small variations in parameters $(\varrho^{1}_n, \varrho^{2}_n, \varrho_n)$ led to large changes in the estimated functions. The primary reason for this observation is that the solution to the estimation and fusion problems involve matrix inversion. The matrices generated for many sets of parameters where ill-conditioned. To circumvent the issue, the values of the parameters were carefully chosen for all $n$ by iterating through idfferent sets of parameter values. 
\section{Conclusion and Future Work}\label{Section 6}
We presented the proof of consistency of the collaborative estimation algorithm. We developed the necessary machinery needed for the same, by : (i) defining parameterized  estimation operators  at the agents and the fusion space, and, studying its asymptotic characteristics (ii) defining the multi-agent estimation and upload operator, multi -agent download operator and studying the properties of these operators. 

As future work we are interested in extending the framework for infinite dimensional knowledge spaces. This requires that the proof of Propositions \ref{Proposition 2}, \ref{Proposition 3} and Lemmas \ref{Lemma 8}, \ref{Lemma 12}  be revisited in the context of infinite dimensional spaces. We are also interested in quantifying knowledge transfer among the agents and understand if agents can learn from each other.
\appendix
\section{Appendix}\label{Appendix}
\subsection{Knowledge Spaces and Operators}
The construction of the knowledge spaces for the individual agents and the fusion space has been discussed in detail in \cite{raghavan2024distributed}. We mention the key results here which have been referenced in the previous sections.
\begin{lemma}\label{Lemma 27}
If $K^{i}(\cdot,\cdot)$ is the reproducing kernel of Hilbert space $H^{i}$,  with norm $||\cdot||_{H^i}$, then $K(x,y)=K^{1}(x,y) + K^{2}(x,y)$ is the reproducing kernel of the space $H = \{f| f= f^1 + f^2 | f^{i} \in H^{i}\}$ with the norm:
\begin{align*}
||f||^2_{H} = \underset{\substack{f^1 + f^2 = f,\\ f^{i} \in H^{i}} }  \min \;\;  ||f^{1}||^2_{H^1} + ||f^{2}||^2_{H^2}.
\end{align*} 
\end{lemma}  
\begin{corollary}\label{Corollary 28}
The uploading operator from Agent $i$'s knowledge space, $H^i$, to the fusion space $H$ is 
\begin{align*}
    &\hat{L}^{i}: H^{i} \to H, \; \hat{L}(f) = f.\\
    &|| \hat{L}^{i} || = \sup\{ || f||_{H} : f\in H^{i}, || f ||_{H^{i}} =1 \}  \leq  1,
\end{align*}
that is  $\hat{L}^{i}(\cdot)$, is linear and is bounded.
\end{corollary}
\begin{lemma}\label{Lemma 29}
Given the RKHS, $(H, \langle \cdot,\cdot \rangle_{H})$, with kernel $K(\cdot,\cdot)$ and the kernels $K^{i}(\cdot,\cdot),\;i=1,2$, such that $K(x,y) = K^{1}(x,y) + K^{2}(x,y)$, we define operators, $\bar{L}^{i}: H \to H$, as
\begin{align*}
\bar{L}^{i}(f)(x) =\langle f(\cdot), K^{i}(\cdot,x) \rangle_{H}, \text{ for}, i=1,2.
\end{align*}
Then, $\bar{L}^{i}$ is linear, symmetric, positive and bounded, $|| \bar{L}^{i} ||\leq 1$.
\end{lemma}
\begin{lemma}\label{Lemma 30}
The linear space, 
\begin{align*}
    \bar{H}^{i} = \{g : g = \sqrt{\bar{L}^i}(f), f\in H \}
\end{align*}
is a RKHS with kernel $K^i$.  $\sqrt{\bar{L}^i}(\cdot)$ establishes an isometric isomorphism between $\mathcal{N}\big(\sqrt{\bar{L}^i}\big)^{\perp}$ and $\bar{H}^{i}$, and the norm, $||f||_{\bar{H}^{i}} = ||g||_{H}$,  where $f = \sqrt{\bar{L}^i} g,  g\in \mathcal{N}\big(\sqrt{\bar{L}^i}\big)^{\perp}$. The downloading operator from the fusion space $H$ to  Agent $i$'s knowledge space, $H^i$, is $\sqrt{\bar{L}^{i}} \circ \Pi_{\mathcal{N}\big(\sqrt{\bar{L}^i}\big)^{\perp}}$. The downloading operator is linear and bounded.
\end{lemma}
For the proofs of Lemma \ref{Lemma 27}, Corollary \ref{Corollary 28}, Lemmas \ref{Lemma 29} and \ref{Lemma 30},  we refer to \cite{raghavan2024distributed}. 
\subsection{Proof of Lemma \ref{Lemma 8}}
In this subsection, we present the proof the the propositions invoked in the proof of Lemma \ref{Lemma 8}.
\begin{proposition}\label{Proposition 31}
\hspace{-5pt} $\{\phi^{i}_{n,1}(\cdot)\}$ is an uniformly bounded equicontinuous sequence of functions on $E^{i}$. 
\end{proposition}
\begin{proof}
Let $(\boldsymbol{\alpha_0};x_0;y_0) \in E^{i}$. Since $E^{i}$ is compact, $\exists \; c_{1} < \infty$ such that $|y| < c_{1}, y \in E^{i}$. Since the sequence $\{\frac{1}{\varrho_{n}}\}$ is bounded, $\exists \; c_{2} < \infty$ such that $|\frac{1}{\varrho_{n}}| \leq c_{2} \forall n$. Since $K^{i}(\cdot,\bar{x}^i_j)$ is uniformly continuous on $E^{i}$ for each $j$, $\mathbf{\bar{K}^{i}}(\cdot)$ is uniformly continuous on $E^{i}$. Thus, $\forall \epsilon > 0$, $\exists \delta_{1} >0$ such that $|| x- x_{0}|| < \delta_1$ implies $|| \mathbf{\bar{K}^{i}}(x) -  \mathbf{\bar{K}^{i}}(x_0) || < \frac{\epsilon}{2 \times c_{1} \times c_{2}} \implies || \frac{\mathbf{\bar{K}^{i}}(x) y}{\varrho^{i}_{n}}- \frac{\mathbf{\bar{K}^{i}}(x_0) y}{\varrho^{i}_{n}}|| <  \frac{\epsilon}{2}$. Since $g(\alpha) = \mathbf{K^{i}} \boldsymbol{\alpha}$ is linear in $\boldsymbol{\alpha}$, $\forall \epsilon > 0$, let $0 < \delta_2 < \frac{\epsilon}{2 \times \lambda_{\max}( \mathbf{K^{i}})}$, then, $|| \boldsymbol{\alpha} - \boldsymbol{\alpha_0} || < \delta_2$ implies $|| \mathbf{K^{i}} \boldsymbol{\alpha} - \mathbf{K^{i}} \boldsymbol{\alpha_0} || = ||  \mathbf{K^{i}}(\boldsymbol{\alpha} -  \boldsymbol{\alpha_0} ) || < \frac{\epsilon}{2}$. Let $\delta = \min(\delta_1, \delta_2)$. Consider any vector in $(\boldsymbol{\alpha_0};x_0;y_0) \in E^{i}$ and let $(\boldsymbol{\alpha};x;y) \in E^{i}$ be such that $|| (\boldsymbol{\alpha};x;y) - (\boldsymbol{\alpha_0};x_0;y_0) || < \delta$. Then $|| \phi^{i}_{n,1}(\boldsymbol{\alpha};x;y) -  \phi^{i}_{n,1}(\boldsymbol{\alpha_0};x_0;y_0) || <  || \frac{\mathbf{\bar{K}^{i}}(x) y}{\varrho^{i}_{n}}- \frac{\mathbf{\bar{K}^{i}}(x_0) y}{\varrho^{i}_{n}}|| + || \mathbf{K^{i}} \boldsymbol{\alpha} - \mathbf{K^{i}} \boldsymbol{\alpha_0} || < \epsilon \; \forall n$. Thus,  $\{\phi^{i}_{n,1}(\cdot)\}$ is uniformly equicontinuous on $E^{i}$. Further, it is clear that $\{\phi^{i}_{n,1}(\cdot)\}$ is uniformly bounded on $E_{i}$, i.e., $\exists c_{\phi^{i}_{n,1}}$ such that $|| \phi^{i}_{n,1}(\boldsymbol{\alpha};x;y) || <  c_{\phi^{i}_{n,1}} \forall  (\boldsymbol{\alpha};x;y) \in E^{i}, \forall n$.
\end{proof}
\begin{proposition}\label{Proposition 32}
The map from $GL_{m}(\mathbb{R}) \to M_{m}(\mathbb{R})$ given by  $M \mapsto M^{-1}$ is continuous in the operator norm topology.
\end{proposition}
\begin{proof}
Let $M \in GL_{n}(\mathbb{R})$. Define, $S= M^{-1} \sum^{\infty}_{n=0} (M -N)^nM^{-n}$, $N \in M_{m}(\mathbb{R})$. The sequence in this definition converges if $||M -N || \times ||M^{-1}|| < 1$ as the space $M_{m}(\mathbb{R})$ is complete. Let $N$  be such that $||M -N || < \frac{1}{||M^{-1}||}$. Then, $S - S(M-N)M^{-1} = M^{-1} \implies S(\mathbb{I}_{m} - (M-N)M^{-1}) = M^{-1} \implies S= M^{-1}(\mathbb{I}_{m} - (M-N)M^{-1})^{-1} = ((\mathbb{I}_{m} - (M-N)M^{-1})M)^{-1} = N^{-1}$. $\mathbb{I}_{m} - (M-N)M^{-1}$ is invertible as $|| \mathbb{I}_{m} - (M-N)M^{-1} || \geq  ||\mathbb{I}_{m}|| - ||(M-N)M^{-1}|| > 0 $. The expansion of $N^{-1}$ gives the bound $|| N^{-1} || \leq \frac{||M^{-1}||}{1 - || M - N|| \; ||M^{-1}|| }$. Thus, $N \in B_{\frac{1}{|| 2  M^{-1}||}}(M)$ is invertible and $|| N^{-1} || \leq 2 ||M^{-1}||$. For such N, $ || N^{-1} - M^{-1} || = ||M^{-1} (M -N) N^{-1} ||\leq 2 ||M^{-1} ||^{2} ||M - N ||$. Given $\epsilon$, let $\delta < \min (\frac{1}{2|| M^{-1}||}, \frac{\epsilon}{2||M^{-1} ||^{2}})$. Then, $|| N - M || < \delta \implies || N^{-1} - M^{-1} || < \epsilon$. Since the map $M \mapsto M^{-1}$ is nonlinear map, by the $\epsilon -\delta$ definition of continuity, the map is continuous.
\end{proof}
\begin{proposition}\label{Proposition 33}
\hspace{-5pt} $\{\phi^{i}_{n,2}(\cdot)\}$ is an uniformly bounded equicontinuous sequence of functions on $E^{i}$. 
\end{proposition}
\begin{proof}
We note that $\mathbf{K^{i}} +\frac{\mathbf{\bar{K}^{i}}(x) \mathbf{\bar{K}^{i^T}}(x) }{\varrho^{i}_{n}}$ need \textit{not} be invertible $\forall n$ since $\mathbf{K^{i}}$ is only positive semidefinite. In case $\mathbf{K^{i}}$ is not invertible, it is replaced with $\mathbf{K^{i}} + c^{i}_{3}\mathbb{I}_{m}$, where $c^{i}_{3}$ is small positive constant, to ensure that the previous sequence of matrices is invertible. Since $\mathbf{\bar{K}^{i}}(\cdot)$ is uniformly continuous on $E^{i}$, given $\epsilon >0$, $\exists \delta >0$ such that $|| x- x_{0} || < \delta$ implies $||\mathbf{\bar{K}^{i}}(x) \mathbf{\bar{K}^{i^T}}(x)  - \mathbf{\bar{K}^{i}}(x_0) \mathbf{\bar{K}^{i^T}}(x_0) || < \frac{\epsilon}{c_{2}}$. Thus, $|| x- x_{0} || < \delta$ implies $|| \mathbf{K^{i}} +\frac{\mathbf{\bar{K}^{i}}(x) \mathbf{\bar{K}^{i^T}}(x) }{\varrho^{i}_{n}} - \mathbf{K^{i}} - \frac{\mathbf{\bar{K}^{i}}(x_0) \mathbf{\bar{K}^{i^T}}(x_0) }{\varrho^{i}_{n}} || = ||\frac{\mathbf{\bar{K}^{i}}(x) \mathbf{\bar{K}^{i^T}}(x)  - \mathbf{\bar{K}^{i}}(x_0) \mathbf{\bar{K}^{i^T}}(x_0)}{\varrho^{i}_{n}}|| < \epsilon, \forall n$. For $x_0 \in E^{i}$, let $\underset{n \in \mathbb{N}} \sup \; || \mathbf{K^{i}} + \frac{\mathbf{\bar{K}^{i}}(x_0) \mathbf{\bar{K}^{i^T}}(x_0)}{\varrho^{i}_{n}} || = c_{4}(x_0) > 0$. Given $\epsilon >0$, let $\delta >0 $ be such that  $|| x- x_{0} || < \delta$ implies $||\frac{\mathbf{\bar{K}^{i}}(x) \mathbf{\bar{K}^{i^T}}(x)  - \mathbf{\bar{K}^{i}}(x_0) \mathbf{\bar{K}^{i^T}}(x_0)}{\varrho^{i}_{n}}|| < \frac{\epsilon}{2 c^2_{4}(x_0)}, \forall n$. From Proposition \ref{Proposition 32} it follows that, for $x$ such that $|| x- x_{0} || < \delta$, $|| \phi^{i}_{n,2}(\boldsymbol{\alpha};x;y) - \phi^{i}_{n,2}(\boldsymbol{\alpha}_0;x_0;y_0)|| \leq 2 c^2_{4}(x_0) \times ||\mathbf{K^{i}} +\frac{\mathbf{\bar{K}^{i}}(x) \mathbf{\bar{K}^{i^T}}(x) }{\varrho^{i}_{n}} - \mathbf{K^{i}} - \frac{\mathbf{\bar{K}^{i}}(x_0) \mathbf{\bar{K}^{i^T}}(x_0) }{\varrho^{i}_{n}}  || <  \epsilon, \forall n$. Due to the compactness of the set $E^i$, we note that $\{\phi^{i}_{n,2}(\cdot)\}$ is uniformly bounded on $E_{i}$, i.e., $\exists c_{\phi^{i}_{n,2}}$ such that $|| \phi^{i}_{n,2}(\boldsymbol{\alpha};x;y) || <  c_{\phi^{i}_{n,2}} \forall  (\boldsymbol{\alpha};x;y) \in E^{i}, \forall n$.
\end{proof}
\begin{proposition}\label{Proposition 34}
$\{\phi^{i}_{n,2}(\cdot)\phi^{i}_{n,1}(\cdot)\}$ is uniformly bounded and equicontinuous sequence of functions on $E^i$.
\end{proposition}
\begin{proof}
Indeed, let $(\boldsymbol{\alpha_0};x_0;y_0) \in E^{i}$. Given $\epsilon >0$, there exists $\delta>0$ such that $||(\boldsymbol{\alpha};x;y) - (\boldsymbol{\alpha_0};x_0;y_0) || < \delta$ implies 
\begin{align*}
    &|| \phi^{i}_{n,1}(\boldsymbol{\alpha};x;y) - \phi^{i}_{n,1}(\boldsymbol{\alpha};x;y) || < \frac{\epsilon}{2 \times c_{\phi^{i}_{n,2}}},\forall n, \\
    &|| \phi^{i}_{n,2}(\boldsymbol{\alpha};x;y) - \phi^{i}_{n,2}(\boldsymbol{\alpha};x;y) || < \frac{\epsilon}{2 \times c_{\phi^{i}_{n,1}}}, \forall n. 
\end{align*}
    Then, $||(\boldsymbol{\alpha};x;y) - (\boldsymbol{\alpha_0};x_0;y_0) || < \delta$, implies the inequalities in \ref{Equation 12}, i.e., the product sequence is uniformly equicontinuous. From Propositions \ref{Proposition 31} and \ref{Proposition 33}, it follows 
    \begin{align*}
    \hspace{-0.7cm}|| \phi^{i}_{n,2}(\boldsymbol{\alpha};x;y)\phi^{i}_{n,1}(\boldsymbol{\alpha};x;y) || < c_{\phi^{i}_{n,2}} c_{\phi^{i}_{n,1}}\forall  (\boldsymbol{\alpha};  x;y) \in E^{i},
     \end{align*}
     $ \forall n, $ i.e the product sequence is uniformly bounded.
\end{proof}
\begin{figure*}
\begin{align}
&|| \phi^{i}_{n,2}(\boldsymbol{\alpha};x;y)\phi^{i}_{n,1}(\boldsymbol{\alpha};x;y) - \phi^{i}_{n,2}(\boldsymbol{\alpha_0};x_0;y_0)\phi^{i}_{n,1}(\boldsymbol{\alpha_0};x_0;y_0)|| = || \phi^{i}_{n,2}(\boldsymbol{\alpha};x;y)\phi^{i}_{n,1}(\boldsymbol{\alpha};x;y) - \phi^{i}_{n,2}(\boldsymbol{\alpha_0};x_0;y_0)\phi^{i}_{n,1}(\boldsymbol{\alpha};x;y) + \nonumber \\
&\phi^{i}_{n,2}(\boldsymbol{\alpha_0};x_0;y_0) \phi^{i}_{n,1}(\boldsymbol{\alpha};x;y) - \hspace{-1pt} \phi^{i}_{n,2}(\boldsymbol{\alpha_0};x_0;y_0)\phi^{i}_{n,1}(\boldsymbol{\alpha_0};x_0;y_0)|| \leq || \phi^{i}_{n,2}(\boldsymbol{\alpha};x;y) ||\; || \phi^{i}_{n,1}(\boldsymbol{\alpha};x;y) - \phi^{i}_{n,1}(\boldsymbol{\alpha_0};x_0;y_0) || + \nonumber \\
&||  \phi^{i}_{n,1}(\boldsymbol{\alpha_0};x_0;y_0)||\; || \phi^{i}_{n,2}(\boldsymbol{\alpha};x;y)  - \phi^{i}_{n,2}(\boldsymbol{\alpha_0};x_0;y_0)|| <  c_{\phi^{i}_{n,2}} \times  \frac{\epsilon}{2 \times c_{\phi^{i}_{n,2}}} + c_{\phi^{i}_{n,1}} \times \frac{\epsilon}{2 \times c_{\phi^{i}_{n,1}}} < \epsilon, \; \forall n. \label{Equation 12} \\
&|| \mathbf{K^{T}} \begin{bmatrix} \boldsymbol{\alpha^{1}} \\ \boldsymbol{\alpha^{2}}  \end{bmatrix} - \phi_{n,1}(\boldsymbol{\alpha^{1}};\boldsymbol{\alpha^{2}}) || = \frac{1}{\varrho_n} || \mathbf{K^{T}}  \begin{bmatrix}  \mathbf{K^{1}}\boldsymbol{\alpha^{1}} \\  \mathbf{K^{2}}\boldsymbol{\alpha^{2}}  \end{bmatrix}|| \geq \frac{1}{\varrho_{n+1}} || \mathbf{K^{T}}  \begin{bmatrix}  \mathbf{K^{1}}\boldsymbol{\alpha^{1}} \\  \mathbf{K^{2}}\boldsymbol{\alpha^{2}}  \end{bmatrix}||
= || \mathbf{K^{T}} \begin{bmatrix} \boldsymbol{\alpha^{1}} \\ \boldsymbol{\alpha^{2}}  \end{bmatrix} - \phi_{n+1,1}(\boldsymbol{\alpha^{1}};\boldsymbol{\alpha^{2}}) ||. \label{Equation 13} \\
&|| \phi_{n,2}(\boldsymbol{\alpha^{1}};\boldsymbol{\alpha^{2}})\phi_{n,1}(\boldsymbol{\alpha^{1}};\boldsymbol{\alpha^{2}}) -  \begin{bmatrix} \boldsymbol{\alpha^{1}} \\ \boldsymbol{\alpha^{2}}  \end{bmatrix} || = || \phi_{n,2}(\boldsymbol{\alpha^{1}};\boldsymbol{\alpha^{2}})\phi_{n,1}(\boldsymbol{\alpha^{1}};\boldsymbol{\alpha^{2}}) - \mathbf{K^{-1}} \phi_{n,1}(\boldsymbol{\alpha^{1}};\boldsymbol{\alpha^{2}}) + \mathbf{K^{-1}} \phi_{n,1}(\boldsymbol{\alpha^{1}};\boldsymbol{\alpha^{2}}) - \mathbf{K^{-1}} \mathbf{K^{T}} \begin{bmatrix} \boldsymbol{\alpha^{1}} \\ \boldsymbol{\alpha^{2}}  \end{bmatrix} || \nonumber \\ 
&\leq || \phi_{n,1}(\boldsymbol{\alpha^{1}};\boldsymbol{\alpha^{2}}) || \; ||  \phi_{n,2}(\boldsymbol{\alpha^{1}};\boldsymbol{\alpha^{2}}) -  \mathbf{K^{-1}} || + ||  \mathbf{K^{-1}}||
|| \phi_{n,1}(\boldsymbol{\alpha^{1}};\boldsymbol{\alpha^{2}}) -  \mathbf{K^{T}} \begin{bmatrix} \boldsymbol{\alpha^{1}}  \\ \boldsymbol{\alpha^{2}}  \end{bmatrix}  || < c_{\phi_{n,1}} \times \frac{\epsilon}{2 \times c_{\phi_{n,1}}} + c_{\phi_{n,2}} \times \frac{\epsilon}{2 \times c_{\phi_{n,2}}} = \epsilon, \label{Equation 14} 
\end{align} 
\vspace{-1cm}
\end{figure*}
\subsection{Proof of Lemma \ref{Lemma 12}}
In this subsection, we present the proof the the propositions invoked in the proof of Lemma \ref{Lemma 12}.
\begin{proposition}\label{Proposition 35}
\hspace{-3pt} $\{\phi_{n,1}(\cdot)\}$ converges uniformly to $\mathbf{K^{T}} \begin{bmatrix} \boldsymbol{\alpha^{1}} \\ \boldsymbol{\alpha^{2}}  \end{bmatrix} $. 
$\{\phi_{n,2}(\cdot)\phi_{n,1}(\cdot)\}$ converges uniformly to $\begin{bmatrix} \boldsymbol{\alpha^{1}} \\ \boldsymbol{\alpha^{2}}  \end{bmatrix}$. 
\end{proposition}
\begin{proof}
\hspace{-5pt} 
From the inequality in \ref{Equation 13}, $\{\bar{\phi}_{n,1}(\boldsymbol{\alpha^{1}};\boldsymbol{\alpha^{2}}) = \hspace{-3pt} ||\mathbf{K^{T}} \begin{bmatrix} \boldsymbol{\alpha^{1}} \\ \boldsymbol{\alpha^{2}}  \end{bmatrix} \\ - \phi_{n,1}(\boldsymbol{\alpha^{1}};\boldsymbol{\alpha^{2}}) ||\}$ is monotone decreasing sequence. Let $E_{n, \epsilon} = \{(\boldsymbol{\alpha^{1}};\boldsymbol{\alpha^{2}}) \in E: \bar{\phi}_{n,1}(\boldsymbol{\alpha^{1}};\boldsymbol{\alpha^{2}}) < \epsilon\}, \epsilon >0$. Since $\bar{\phi}_{n,1}(\cdot)$ is continuous, $E_{n, \epsilon}$ is open. $\{E_{n, \epsilon}\}$ is an ascending sequence of open sets, i.e., $E_{n} \subset E_{n+1}$. Since $\{\bar{\phi}_{n,1}(\cdot)$ converges pointwise to zero, $\{E_{n, \epsilon}\}$ is an open cover for $E$. Since $E$ is compact, there exists a finite subcover. Since $\{E_{n, \epsilon}\}$ is ascending, the maximum index from the finite subcover is a cover too, i.e., $\exists N_{\epsilon}$ such that $E_{N_{\epsilon}} = E$. We note that $\{\phi_{n,1}(\cdot)\}$ is uniformly bounded, i.e., $\exists c_{\phi_{n,1}}$ such that, $|| \phi_{n,1}(\boldsymbol{\alpha^{1}};\boldsymbol{\alpha^{2}}) || < c_{\phi_{n,1}} \forall  (\boldsymbol{\alpha^{1}};\boldsymbol{\alpha^{2}}) \in E, \forall n$. Assuming $\mathbf{K}$ is invertible,  we note that $\{\phi_{n,2}(\cdot)\}$ converges uniformly to $\mathbf{K^{-1}}$ since it is independent of $(\boldsymbol{\alpha^{1}};\boldsymbol{\alpha^{2}})$. Further it is uniformly bounded by $\mathbf{K^{-1}}= c_{\phi_{n,2}}$.

Given $\epsilon >0$, there exists $N_{\epsilon}$ such that $|| \phi_{n,1}(\boldsymbol{\alpha^{1}};\boldsymbol{\alpha^{2}}) - \mathbf{K^{T}} \begin{bmatrix} \boldsymbol{\alpha^{1}} \\ \boldsymbol{\alpha^{2}}  \end{bmatrix} || < \frac{\epsilon}{2 \times c_{\phi_{n,2}}}, \forall n \geq N_{\epsilon}, \forall (\boldsymbol{\alpha^{1}};\boldsymbol{\alpha^{2}}) \in E $ and $|| \phi_{n,2}(\boldsymbol{\alpha^{1}};\boldsymbol{\alpha^{2}}) - \mathbf{K^{-1}} || < \frac{\epsilon}{2 \times c_{\phi_{n,1}}}, \forall n \geq N_{\epsilon}$. Thus for $n \geq N_{\epsilon}$, inequality \ref{Equation 14} holds $\forall  (\boldsymbol{\alpha^{1}};\boldsymbol{\alpha^{2}}) \in E $. This implies the second claim of this proposition. 
\end{proof}
\bibliographystyle{model1-num-names}
\bibliography{biblio}
\end{document}